%% file: main.tex
\documentclass{article}
\PassOptionsToPackage{numbers,compress}{natbib}
\usepackage[preprint]{icml2026}

\input{headers/header-lqa-head}

\input{headers/header}

\input{headers/header-lqa-tail}

\newcommand{\OurTitle}{A Unified Framework for LLM Watermarks}

\icmltitlerunning{\OurTitle}

\begin{document}

\twocolumn[
  \icmltitle{\OurTitle}

  \icmlsetsymbol{equal}{*}

  \begin{icmlauthorlist} 
    \icmlauthor{Thibaud Gloaguen}{eth}
    \icmlauthor{Robin Staab}{eth}
	\icmlauthor{Nikola Jovanovi\'c}{eth}
	\icmlauthor{Martin Vechev}{eth}
  \end{icmlauthorlist}

  \icmlaffiliation{eth}{Department of Computer Science, ETH Zürich}

  \icmlcorrespondingauthor{Thibaud Gloaguen}{thibaud.gloaguen@inf.ethz.ch}

  \icmlkeywords{Machine Learning, ICML}

  \vskip 0.3in
]

\printAffiliationsAndNotice{}

\begin{abstract}
	\input{abstract.tex}
\end{abstract}

\input{src/introduction.tex}

\input{src/related_work.tex}

\input{src/method.tex}

\input{src/evaluation.tex}

\input{src/conclusion.tex}

\message{^^JLASTBODYPAGE \thepage^^J}

\clearpage

\input{src/impact_statement.tex}

\bibliographystyle{icml2026}
\bibliography{references}
\vfill
\clearpage

\message{^^JLASTREFERENCESPAGE \thepage^^J}

\ifincludeappendixx
	\newpage
	\appendix
	\onecolumn
	\crefalias{section}{appendix}
	\crefalias{subsection}{appendix} 
	\include{appendix}

\fi

\end{document}

%% file: headers/header-lqa-head.tex
\usepackage[utf8]{inputenc} %

\usepackage{graphicx}

\usepackage[T1]{fontenc}

\usepackage[hidelinks,colorlinks]{hyperref}       %
\usepackage{url}            %
\usepackage{booktabs}       %
\usepackage{amsfonts}       %
\usepackage{nicefrac}       %
\usepackage{microtype}      %

\input{headers/lqa/overfull}

\input{headers/lqa/comments}

\input{headers/lqa/appendix}

\input{headers/lqa/abbreviations}

\input{headers/lqa/acronyms}

\input{headers/lqa/listing}

%% file: headers/lqa/overfull.tex
\IfFileExists{headers/config/showoverfull.config}{
	\overfullrule=1cm
}{
}

%% file: headers/lqa/comments.tex
\usepackage{marginnote}

\usepackage[backgroundcolor=none,linecolor=red,textsize=footnotesize]{todonotes}

%% file: headers/lqa/appendix.tex
\usepackage{etoolbox}

\newbool{includeappendix}
\setbool{includeappendix}{true} %
\IfFileExists{headers/config/noappendix.config}{
	\setbool{includeappendix}{false}
}{}

\newif\ifincludeappendixx
\ifbool{includeappendix}{
	\includeappendixxtrue
}{
	\includeappendixxfalse
}

\usepackage{xr} %
\usepackage{filecontents}

\ifbool{includeappendix}{}{
	\input{appendix-labels-loader}

	\externaldocument{appendix-labels}
}

%% file: headers/lqa/abbreviations.tex
\usepackage{xspace}

\newcommand{\eg}{e.g., }
\newcommand{\ie}{i.e., }

\newcommand{\llama}{\textsc{Llama3.1-8B}}
\newcommand{\qwen}{\textsc{Qwen3-30B}}
\newcommand{\mistral}{\textsc{Ministral-3-14B}}

%% file: headers/lqa/acronyms.tex
\usepackage{acro} %

%% file: headers/lqa/listing.tex
\usepackage{listings}

\usepackage{textcomp}

\usepackage{xcolor}

\usepackage[scaled=0.8]{beramono}

\definecolor{ckeyword}{HTML}{7F0055}
\definecolor{ccomment}{HTML}{3F7F5F}
\definecolor{cstring}{HTML}{2A0099}

\lstdefinestyle{numbers}{
	numbers=left,
	framexleftmargin=20pt,
	numberstyle=\tiny,
	firstnumber=auto,
	numbersep=1em,
	xleftmargin=2em
}

\lstdefinestyle{layout}{
	frame=none,
	captionpos=b,
}

\lstdefinestyle{comment-style}{
	morecomment=[l]//,
	morecomment=[s]{/*}{*/},
	commentstyle={\color{ccomment}\itshape},
}

\lstdefinestyle{string-style}{
	morestring=[b]",%
	morestring=[b]',%
	stringstyle={\color{cstring}},
	showstringspaces=false,%
}

\lstdefinestyle{keyword-style}{
	keywordstyle={\ttfamily\bfseries},
	morekeywords={
		function,
		constructor,
		int,
		bool,
		return,
		returns,
		uint
	},
	morekeywords = [2]{},
	keywordstyle = [2]{\text},
	sensitive=true,
}

\lstdefinestyle{input-encoding}{
	inputencoding=utf8,
	extendedchars=true,
	literate=
	{ℝ}{$\reals$}1%
	{→}{$\rightarrow$}1%
	{α}{$\alpha$}1%
	{β}{$\beta$}1%
	{λ}{$\lambda$}1%
	{θ}{$\theta$}1%
	{ϕ}{$\phi$}1%
}

\lstdefinestyle{escaping}{
	moredelim={**[is][\color{blue}]{\%}{\%}},
	escapechar=|,
	mathescape=true
}

\lstdefinestyle{default-style}{
	basicstyle=\fontencoding{T1}\ttfamily\footnotesize,
	style=numbers,
	style=layout,
	style=comment-style,
	style=string-style,
	style=keyword-style,
	style=input-encoding,
	style=escaping,
	tabsize=2,
	upquote=true
}

\lstdefinelanguage{BASIC}{
	language=C++,
	style=default-style
}[keywords,comments,strings]%

%% file: headers/header.tex
\usepackage[utf8]{inputenc} %
\usepackage[T1]{fontenc}    %
\usepackage{caption}
\usepackage{enumerate}
\usepackage{url}            %
\usepackage{booktabs}       %
\usepackage{amsfonts}       %
\usepackage{nicefrac}       %
\usepackage{microtype}      %
\usepackage{xcolor}         %
\usepackage{tikz,booktabs,multirow,enumitem,bm}
\usepackage{colortbl}
\usepackage{tabularx}
\usepackage{subcaption}
\usepackage{wrapfig}
\usepackage{tabulary}
\usepackage{amsmath,amsthm,amsfonts, bbm}
\usepackage{makecell}
\usepackage{rotating}
\usepackage{array}
\usepackage{siunitx}
\usepackage[most]{tcolorbox}
\usepackage{CJKutf8}
\usepackage{pifont}

\input{headers/math_commands.tex}

\newcolumntype{x}[2]{S[table-format=#1.#2,table-auto-round]}
\newcolumntype{y}[2]{>{\small} S[table-format=#1.#2,table-auto-round]}

\definecolor{hyperlinkblue}{HTML}{0000AA}
\hypersetup{citecolor=hyperlinkblue} %

\definecolor{red}{HTML}{FF0000}
\definecolor{drop}{HTML}{2596be}
\definecolor{anon}{HTML}{10a37f}
\definecolor{devil}{HTML}{84cf9d}

\lstdefinestyle{mystyle}{
    escapechar=\#,
    breaklines=true,
    basicstyle=\scriptsize\ttfamily,
    numbers=none,
    language={},
    framextopmargin=0pt,
    framexbottommargin=0pt,
    breakindent=0pt,
    showspaces = false,
    keywordstyle=\bfseries,
    showstringspaces=false,
    columns=fullflexible,
    morekeywords={Style, Consistency, Accuracy, Ethics, Score}
}

\newtcblisting{prompt}[2][]{
    arc=3pt, outer arc=3pt,
    width=.9\linewidth,
    left=1mm,
    top=0mm,
    bottom=0mm,
    title=#2, 
    colback=gray!5!white,
    colframe=black!75!black,
    fonttitle=\bfseries,
    listing only, 
    listing options={style=mystyle},
    breakable,
    #1
}

\newtcblisting{response}[2][]{
    arc=3pt, outer arc=3pt,
    width=.9\linewidth,
    left=1mm,
    top=0mm,
    bottom=0mm,
    title=#2, 
    colback=devil!5!white,
    colframe=devil!50!black,
    fonttitle=\bfseries,
    listing only, 
    listing options={style=mystyle,deletekeywords={Style}},
    breakable,
    #1
}

\newtcblisting{inference}[2][]{
    arc=3pt, outer arc=3pt,
    width=.9\linewidth,
    left=1mm,
    top=0mm,
    bottom=0mm,
    title=#2, 
    colback=red!5!white,
    colframe=red,
    fonttitle=\bfseries,
    listing only, 
    listing options={style=mystyle},
    breakable,
    #1
}

\newtcblisting{anonymize}[2][]{
    arc=3pt, outer arc=3pt,
    width=.9\linewidth,
    left=1mm,
    top=0mm,
    bottom=0mm,
    title=#2, 
    colback=anon!5!white,
    colframe=anon,
    fonttitle=\bfseries,
    listing only, 
    listing options={style=mystyle},
    breakable,
    #1
}

\lstdefinelanguage{QA}{
  morekeywords={Prompt,Base,Watermarked},
  sensitive=false,
}

\lstdefinestyle{qaStyle}{
  language=QA,
  basicstyle=\normalfont\small,      %
  keywordstyle=\bfseries\color{teal},%
  stringstyle=\color{black},         %
  numbers=none,                       %
  showstringspaces=false,            %
  frame=single,                      %
  rulecolor=\color{gray!50},         %
  backgroundcolor=\color{gray!10},   %
  columns=fullflexible,              %
  xleftmargin=1em, xrightmargin=1em, %
  aboveskip=1em, belowskip=1em       %
}

\definecolor{outerbg}{HTML}{F6F6F6}      %
\definecolor{outerframe}{HTML}{D0D0D0}   %

\definecolor{modelAback}{HTML}{DD8452}   %
\definecolor{modelAframe}{HTML}{DD8452}  %

\definecolor{modelBback}{HTML}{4C72B0}   %
\definecolor{modelBframe}{HTML}{4C72B0}  %

\newtcolorbox{promptbox}[1][]{%
  breakable,
  colback=white,
  colframe=black!25,
  fonttitle=\bfseries\small,
  coltitle=black!80,
  title=User~Prompt~{#1},
  boxsep=5pt,
  top=4pt,bottom=4pt,
  arc=1mm,
}

\newtcolorbox{modelbox}[3][]{%
  breakable,
  colback=#2!75,
  colframe=#3!85,
  coltitle=#3!10!black,
  fonttitle=\bfseries\small,   %
  title={#1},
  boxsep=5pt,
  top=4pt,bottom=4pt,
  arc=1mm,
}

%% file: headers/math_commands.tex
\usepackage{amsmath,amsfonts,bm}
\usepackage{amsthm}
\usepackage{amssymb}
\usepackage{mathtools}

\def\1{\bm{1}}

\DeclareMathAlphabet{\mathsfit}{\encodingdefault}{\sfdefault}{m}{sl}
\SetMathAlphabet{\mathsfit}{bold}{\encodingdefault}{\sfdefault}{bx}{n}

\newcommand{\KL}{D_{\mathrm{KL}}}

\DeclareMathOperator*{\argmax}{arg\,max}

%% file: headers/header-lqa-tail.tex
\input{headers/lqa/references}

\theoremstyle{plain}
\newtheorem{theorem}{Theorem}[section]

\theoremstyle{definition}

\theoremstyle{remark}

%% file: headers/lqa/references.tex
\renewcommand{\S}{Sec.~}

\usepackage[capitalize,noabbrev]{cleveref}

\crefformat{section}{\S#2#1#3}

\crefrangeformat{section}{\S#3#1#4\crefrangeconjunction\S#5#2#6}

\crefmultiformat{section}{\S#2#1#3}{\crefpairconjunction\S#2#1#3}{\crefmiddleconjunction\S#2#1#3}{\creflastconjunction\S#2#1#3}

\newcommand{\crefrangeconjunction}{--}

\crefname{listing}{Lst.}{listings}
\crefname{line}{Lin.}{Lin.}
\crefname{appendix}{App.}{App.}

\newcommand{\appref}[1]{%
	\ifbool{includeappendix}{\cref{#1}}{the appendix}%
}
\newcommand{\Appref}[1]{%
	\ifbool{includeappendix}{\cref{#1}}{The appendix}%
}

%% file: abstract.tex
LLM watermarks allow tracing AI-generated texts by inserting a detectable signal into their generated content. Recent works have proposed a wide range of watermarking algorithms, each with distinct designs, usually built using a bottom-up approach. Crucially, there is no general and principled formulation for LLM watermarking.

In this work, we show that most existing and widely used watermarking schemes can in fact be derived from a principled constrained optimization problem. Our formulation unifies existing watermarking methods and explicitly reveals the constraints that each method optimizes. In particular, it highlights an understudied quality-diversity-power trade-off. At the same time, our framework also provides a principled approach for designing novel watermarking schemes tailored to specific requirements. For instance, it allows us to directly use perplexity as a proxy for quality, and derive new schemes that are optimal with respect to this constraint. Our experimental evaluation validates our framework: watermarking schemes derived from a given constraint consistently maximize detection power with respect to that constraint.

%% file: src/introduction.tex
\section{Introduction}
\label{sec:introduction}

Large Language Model (LLM) watermarking algorithms modify the sampling procedure of the underlying LLM to introduce a detectable signal in the generated text~\citep{kgw1}.
Importantly, prior watermarking works have proposed a range of seemingly unrelated approaches~\citep{kgw1,synthid,aar,kth,undetectable} for how to modify the sampling procedure.
Usually, given a specific motivation (e.g., distortion-freeness or power), they first propose a new watermark design, prove that it satisfies some (theoretical) properties, and experimentally validate it.
As a result, it becomes difficult to compare existing watermarking schemes at a fundamental level, relate them to one another, or design new methods in a principled manner.

\paragraph{This work:}
In this work, we address this issue by introducing a new unified framework for LLM watermarking.
In particular, we show that the construction of a watermarking algorithm can be phrased as a constrained optimization problem: the goal is to maximize the detectability of the inserted watermark signal, while limiting its impact on the quality of the generated text.
By formalizing and analyzing this optimization problem, we find that a number of popular schemes (\eg Red-Green watermarks, AAR/KTH, SynthID) can be formulated as solutions to our optimization problem when instantiated with specific constraints.

Importantly, the explicit nature of the constraints in our framework further allows us to better understand and relate existing schemes and provides a principled approach to the design of new, more powerful ones.
As we show in our empirical evaluation, our framework provides a natural and practical way to design and instantiate powerful schemes tailored to specific constraints such as quality preservation.
Lastly, our framework allows us to theoretically identify and experimentally validate a previously understudied fundamental trade-off in LLM watermarks: quality-diversity-power, \ie the trade-off between the quality of the watermarked text, the output diversity under the watermarking scheme, and the overall watermarking signal strength.

\input{figures/overview/explain_figure.tex}

\paragraph{Main contributions}
Our main contributions are:
\begin{itemize}
    \item We formulate LLM watermarking as a general and principled constrained optimization problem (\cref{sec:method_theory}).
    \item We show that our framework captures a number of prior watermarking algorithms (\cref{sec:method_practice:existing}).
    \item We instantiate novel watermark schemes that are optimal with respect to a particular objective (\cref{sec:method_practice:new}).
    \item We empirically show how each scheme (including those from prior work) is optimal with respect to its specific constraint (\cref{sec:evaluation:metric}), validate the newly derived schemes (\cref{sec:evaluation:detectability}), and measure the quality--diversity--power trade-off (\cref{sec:evaluation:diversity}). Our code is available \href{https://github.com/eth-sri/unified-watermarking}{here}.
\end{itemize}

%% file: figures/overview/explain_figure.tex
\begin{figure*}[t]
    \centering
    \includegraphics[width=\textwidth]{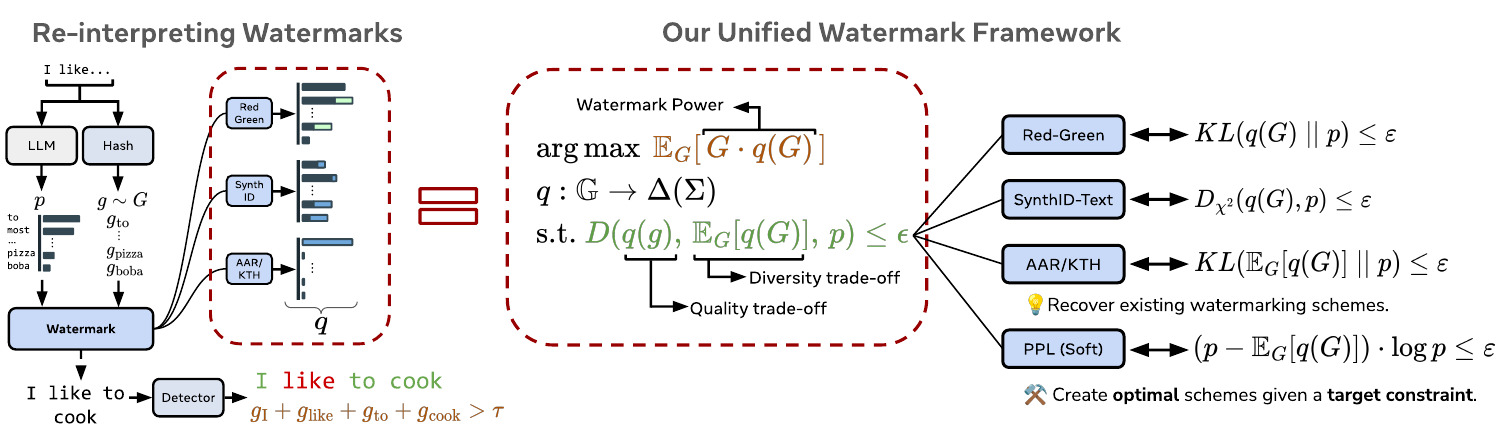}
    \caption{\textbf{Overview of Our Framework.} We find that most prior watermarks can be viewed from the following angle (\emph{left}): given a next-token probability distribution $p$ and token scores $g$, pseudorandomly sampled from $G$, they compute a watermarked probability distribution $q(g)$. 
    A text is considered watermarked if the sum of pseudorandom scores is above a given threshold.
    Therefore, we frame watermarking as a constrained optimization problem (\emph{middle}): maximizing the expected score while controlling the watermark distortion.
    In particular, the constraint balances the watermark quality and diversity. 
    This formulation captures most existing prior watermarks, and enables designing new optimal schemes with respect to a given constraint (\emph{right}).}
    \label{fig:main}
\end{figure*}

%% file: src/related_work.tex
\section{Background and Related Work}
\label{sec:related_work}

We now introduce the necessary background on LLM watermarks, including all relevant related work.
We defer specific details of each watermarking scheme to \cref{app:algorithms}.

\paragraph{LLM Watermarks}
The goal of LLM watermarks is to insert a detectable signal in the text generated by the LLM.
At each generation step, watermarks use a hash of the previously generated context to seed a pseudorandom number generator and attribute a pseudorandom score to each token in the vocabulary. 
We call this the \emph{hashing} mechanism. %
Using those sampled pseudorandom scores, the \emph{sampling mechanism} modifies the original next-token probability distribution provided by the model into a watermarked next-token probability distribution.
The next token is then sampled according to the watermarked distribution.
Then, given a text, the \emph{watermark detector} reconstructs the pseudorandom score for each token in the text and computes a p-value.
The p-value represents the probability of observing such a result or a more extreme one under the assumption that the sequence was sampled independently of the scores. 
In other words, if the text is human-generated, the respective scores are random and the p-value will be high: the text is not watermarked.
By contrast, it is statistically unlikely that many tokens have consistently high scores, allowing us to deduce that such texts are watermarked.

Importantly, the detector should be \emph{model-free}, \ie computing the p-value should not require access to the model and should be computationally efficient.
Indeed, most practical scenarios assume that detection is run at internet scale, where computational efficiency is required.

\paragraph{Watermark Sampling Mechanisms}
Prior works have proposed a variety of sampling mechanisms, trying to balance power with impact on quality.
Popular approaches include \citet{kgw1}, which biases the logits of the next-token distribution according to binary scores, \citet{synthid} leverages tournament sampling to favor higher scoring tokens, \citet{aar} deterministically samples the highest scoring token according to the logarithm of the next-token probability distribution and a Gumbel score, \citet{kth,undetectable} use inverse-transform sampling, and~\citet{dipmark, mcmark} reweight the next-token distribution according to the scores'~ranking.
While those sampling mechanisms are seemingly unrelated, our work shows that most of them can be derived from a unified family of constrained optimization problems~(\cref{sec:method_practice:new}). 

\paragraph{Optimal Watermark Design}
Several prior works have~attempted to design LLM watermarks, or some of their components, from a theoretical perspective.
For the hashing~mechanism, \citet{semstamp, semamark, simkey} propose using a semantic-aware hash to increase the watermark detectability after paraphrasing, and \citet{bileve, scrubbing_def1} improve the resilience of the hashing mechanism against spoofing attacks.
\citet{optimal_statistical_wm, statistical_wm, likelihood_wm_detection} study the watermark detector to find the shape of the optimal detector with respect to the test power.

More closely related to our work, \citet{optimizing_wm, optimized_couplings_wm} study the sampling mechanism that consists of an exponential tilting of the next-token probability distribution with binary scores (a similar setting to~\citet{kgw1}) and find a closed-form formula to compute the watermarked distribution that maximizes power under a log-likelihood constraint.
\citet{gloaguen2025watermarkingdiffusionlanguagemodels} frames Red-Green watermarks as a constrained optimization problem to derive a watermarking rule for diffusion language models.
\citet{ot_watermark} frame watermarking as an optimal transport problem, from which they derive distortion-free schemes that deterministically (according to the pseudorandom scores and $p$) sample each token, recovering the watermarking scheme from~\citet{aar}.
Lastly, \citet{dawa} jointly optimize the detector and the sampling mechanism, yielding a powerful yet non-robust watermark.
Based on their theoretical insights, they instantiate a distortion-free scheme that introduces dependence between the score and the next-token~distribution.

Unlike our approach, existing theoretical frameworks focus solely on finding optimal components for specific LLM watermarks and therefore are not capable of capturing and extending a wide range of existing watermarking schemes.

%% file: src/method.tex
\section{Our Framework}
\label{sec:method_theory}

In this section, we describe our unified framework for LLM watermarks, illustrated in \cref{fig:main}, and motivate its design.

\input{tables/table_overview.tex}

\paragraph{Setting}
We recall the main components of watermarks (\cref{fig:main}, left): the hashing mechanism to compute pseudorandom scores, the sampling mechanism that, given the next-token probability distribution, computes a watermarked probability distribution, and the detector that, given a text, assesses whether it is watermarked.
Our goal is to design a principled sampling mechanism, given a fixed hashing mechanism and a fixed detector. %
Let $\Sigma$ be the finite vocabulary, $p \in \Delta(\Sigma)$ be the next-token probability distribution, $G$ the random variable representing the pseudorandom scores with $\mathbb{G} \subset \mathbb{R}^{|\Sigma|}$ as support, and $g \sim G$ a sample from $G$ computed from the context via the hashing mechanism.
At every generation step, we want to find a watermarked probability distribution $q(g)$ such that later, without knowledge of $p$, we can distinguish whether a sequence of tokens was sampled from $p$ or $q(g)$.

\paragraph{Optimization Objective}
A priori, finding such a distribution $q(g)$ is a complex task.
If our detection can depend on $p$ and $q(g)$, the Neyman–Pearson lemma provides a formula for the best possible detector.
Hence, finding $q(g)$ can be simplified to maximizing the power of this optimal detector.
However, in practice, we want model-free detection (\cref{sec:related_work}), \ie the detector should depend on neither $p$ nor $q(g)$.
From~\citet{dawa}, we know that an optimal model-free detector is intractable.
In the literature, popular schemes often use the sum of individual token scores as their detection statistic.
This suggests that, at each generation step, to maximize the power of the detector one should sample the next-token with the highest possible score.
In our notation, this corresponds to solving
\begin{equation} \label{eq:energy_function_single_instance}
    q^*(g) = \argmax_{q \in \Delta(\Sigma)} \mathbb{E}_{\Omega \sim q} [g_\Omega] \coloneq \argmax_{q \in \Delta(\Sigma)} g \cdot q.
\end{equation}
Solving \cref{eq:energy_function_single_instance} for all $g\in\mathbb{G}$ is then equivalent 
to solving for all $q: \mathbb{G} \rightarrow \Delta(\Sigma)$,
\begin{equation} \label{eq:energy_function}
    q^* = \argmax_{q: \mathbb{G} \rightarrow \Delta(\Sigma)} \mathbb{E}_G [G \cdot q(G)].
\end{equation}

\paragraph{Adding Constraints}
Yet, purely optimizing \cref{eq:energy_function} gives a trivial solution: for all $g$, $q^*(g)$ corresponds to greedy sampling according to the scores $g$.
This is clearly undesirable, as such a watermark would ignore and thus heavily distort the original next-token distribution $p$, leading to low-quality generations.
We therefore introduce a constraint to \cref{eq:energy_function} that controls how distortionary the watermark can be.
With this, our watermark framework becomes
\begin{equation}\label{eq:constrained_optimization_problem}
\begin{aligned}
    q^* = & \argmax_{q: \mathbb{G} \rightarrow \Delta(\Sigma)}
    \quad \mathbb{E}_G [G \cdot q(G)] \\
    & \text{s.t. } \forall g \in \mathbb{G}, \quad
    D\left(q(g), \mathbb{E}_{G}[q(G)],p\right) \le \varepsilon.
\end{aligned}
\end{equation}

The constraint $D \colon \Delta(\Sigma)^3 \rightarrow \mathbb{R}$ allows to capture both distortion given a fixed $g$ through $q(g)$, and distortions in expectation through $\mathbb{E}_{G}[q(G)]$.
For instance, we can choose as $D$ a KL constraint $\text{KL}(q(g)||p)$ or a distortion-free constraint $\text{KL}(\mathbb{E}_G[q(G)]||p)$.
With our framework, as we show in \cref{sec:method_practice}, designing a watermarking sampling mechanism now becomes equivalent to choosing a meaningful metric $D$ and then solving \cref{eq:constrained_optimization_problem}.

\paragraph{Hard and Soft Constraints}
We define $D$ such that it depends on both $q(g)$ and $\mathbb{E}_{G}[q(G)]$.
We consider imposing a constraint on $q(g)$ as a \emph{hard constraint}, as the watermark cannot significantly distort the probability distribution given a fixed $g$.
Likewise, imposing a constraint on $\mathbb{E}_{G}[q(G)]$ is a \emph{soft constraint}.
A single instantiation of $q(g)$ can heavily distort the distribution, but on average, the constraint holds.
This means that, in practice, there are two different kinds of distortions when it comes to watermarking.
We find in \cref{sec:evaluation:diversity} that soft-constraint watermarks are more powerful and better preserve quality, but exhibit very low diversity.
In particular, we show in \cref{theorem:distortion_free_argmax} that (optimal) watermarks with only soft constraints are necessarily deterministic (\ie given a fixed prompt, the watermarked LLM always samples the same text).
In contrast, hard constraints tend to better preserve diversity at the cost of quality.
Interestingly, although the separation of hard and soft constraints is not strictly necessary in our framework, we found that all existing watermarking schemes from \cref{sec:method_practice:existing} can be expressed using only one of the two. 
This explicit separation provides a novel and more principled perspective on the watermark quality-diversity-power trade-off.

\paragraph{Penalized Formulation}
In practice, solving the constrained optimization problem from \cref{eq:constrained_optimization_problem} at each generation step can prove to be too costly.
Hence, for practical reasons, most schemes from \cref{sec:method_practice} actually correspond to the penalized version of the optimization problem,
\begin{align} \label{eq:penalized_optimization_problem}
    \argmax_{q: \mathbb{G} \rightarrow \Delta(\Sigma)} \mathbb{E}_G \!\left[G \!\cdot \!q(G) - \frac{1}{\delta} D\left(q(G), \mathbb{E}_{G}[q(G)], p\right)\right]\!.
\end{align}
With the penalized version, assuming the existence of an interior point, the shape of the optimal solution $q^*$ is the same as the constrained one: for all $\varepsilon > 0$, there exists $\delta \in \mathbb{R}_+$ such that \cref{eq:constrained_optimization_problem} and \cref{eq:penalized_optimization_problem} have the same solution.
Hence, the penalized version corresponds to parametrizing the scheme with $\delta$ instead of $\varepsilon$.
For clarity and consistency, in \cref{sec:method_practice} when we instantiate a watermark sampling mechanism, we parametrize it by $\varepsilon$ if we solve the constrained formulation (\cref{eq:constrained_optimization_problem}) in practice, and by $\delta$ if we solve the penalized formulation (\cref{eq:penalized_optimization_problem}). 

\section{Applying our Framework}
\label{sec:method_practice}

In this section, we show how our unified framework naturally describes both existing and new schemes, as illustrated in \cref{fig:main} (right).
In \cref{sec:method_practice:existing}, we show how our framework captures some popular schemes from the literature and even allows us to extend them.
For instance, AAR/KTH, which are widely seen as fixed sampling schemes, can be seen as solutions to \cref{eq:constrained_optimization_problem} with a distortion-free constraint, or a tournament layer in SynthID can be seen as a solution to \cref{eq:constrained_optimization_problem} with a $\chi^2$-distance constraint.
At the same time, in \cref{sec:method_practice:new}, we show how to instantiate entirely novel schemes given particular requirements.
We summarize all our results in \cref{tab:overview} and defer all proofs to \cref{app:proofs}.

\subsection{Capturing Existing Watermarking Sampling Mechanisms}
\label{sec:method_practice:existing}

We now demonstrate how some of the most popular sampling mechanisms from prior works~\citep{kgw1,aar,kth,synthid} can be derived and extended when instantiated with a specific constraint.

\paragraph{Red-Green Watermarks}
One of the most prominent schemes in the literature is the Red-Green watermark (KGW) from~\citet{kgw1}.
At each step $t$, the Red-Green watermark pseudo-randomly partitions the vocabulary into a set of $\gamma|\Sigma|$ green tokens and the red tokens.
Then, it adds a constant $\delta > 0$ to the logits of tokens in the green set, effectively boosting the probability of green tokens.
The detector thresholds the number of green tokens in a text to determine whether it is watermarked.

To capture Red-Green watermarks, we let the entries of $G$ be correlated Bernoulli random variables, ensuring it has exactly $\gamma|\Sigma|$ ones.
Then, we use the following constraint,
\begin{equation} \label{eq:kgw_constraint}
\forall g\in \mathbb{G}, \, D_{\text{R-G}}(q(g), \mathbb{E}_G[q(G)], p) := \text{KL}(q(g) || p).
\end{equation}
Given $\delta > 0$ as the penalization constant, we show in \cref{theorem:red_green_watermarks} that the solution to our optimization problem from \cref{eq:penalized_optimization_problem} is given by
\begin{equation} \label{eq:exponential_shift}
    \forall g \in \mathbb{G}, \, q_{\text{R-G}}(g) \propto p \exp(\delta g).
\end{equation}
In the logit space, we have, up to an additive constant,
\begin{equation} \label{eq:exponential_shift_log}
    \forall g \in \mathbb{G}, \, \log q_{\text{R-G}}(g) = \log p + \delta g.
\end{equation}
This expression exactly matches the Red-Green watermark algorithm described above: if a token $u \in \Sigma$ is green ($g_u = 1$), we add $\delta$ to its logit.
Hence, we have shown that Red-Green watermarks can be naturally derived from our framework as solutions to \cref{eq:penalized_optimization_problem} using a (hard) KL constraint, and we find in practice (\cref{sec:evaluation:metric}) that Red-Green watermarks are indeed Pareto-optimal with respect to the detectability--KL constraint trade-off.

At the same time, because Red-Green watermarks only use a hard constraint, \ie the constraint depends only on $q(g)$ and not on the distribution of $G$.
Hence, as we show in \cref{theorem:red_green_watermarks}, \cref{eq:exponential_shift} is the solution of \cref{eq:penalized_optimization_problem} with the constraints from \cref{eq:kgw_constraint} for any distribution of $G$.
This allows defining new watermarking schemes that we evaluate in \cref{app:additional_eval:distrib_ablation}.

\paragraph{AAR and KTH Watermarks}
Both \citet{aar} and \citet{kth} (EXP-edit) used the Gumbel-max trick to build a distortion-free watermark.
They associate with each token i.i.d. Gumbel scores and sample the next token as the argmax of the Gumbel scores and the token probability. 
In expectation over the Gumbel scores, such sampling is distributed exactly as the original probability.

To capture AAR, we use a soft-constraint with the entries of $G$ being i.i.d. Gumbel random variables, for all $g \in \mathbb{G}$,
\begin{equation} \label{eq:aar_constraint}
    D_{\text{AAR}}(q(g), \mathbb{E}_G[q(G)], p) := \text{KL}(\mathbb{E}_G[q(G)] || p).
\end{equation}
Given $\varepsilon = 0$ (\ie enforcing distortion-freeness), \cref{theorem:aar_watermarks} gives us the following solution
\begin{equation} \label{eq:aar_mechanism}
    \forall u \in \Sigma, q_{\text{AAR}}(g)_u := \mathbbm{1}\{ u = \argmax_{v \in \Sigma} (g_v + \log p_v)\}.
\end{equation}
This is exactly the AAR/KTH watermark, which our framework recovers when enforcing the distortion-free constraint with Gumbel scores $G$.
This also shows that, when using Gumbel scores with sum-based detectors, AAR/KTH are the most powerful distortion-free schemes.  

Importantly, with \cref{eq:aar_constraint} and $\varepsilon > 0$, our framework extends AAR/KTH to a (more powerful) distortionary variant.
Indeed, we show in \cref{theorem:aar_watermarks} that, when using $D_{\text{AAR}}$, we obtain the following sampling mechanism:
\begin{equation}
    q_{\text{AAR}}(g)_u := \mathbbm{1}\left\{ u = \argmax_{v \in \Sigma} \left(g_v + 
    \frac{1}{1+\delta} \log p_v\right)\right\},
\end{equation}
where $u \in \Sigma$ and $\delta \ge 0$.
We evaluate this extension in \cref{sec:evaluation}, and find it to be a powerful distortionary scheme.

Our framework also generalizes the AAR/KTH mechanism to any distribution of $G$.
For instance, when using a log-normal distribution for $G$, it directly corresponds to the HeavyWater scheme~\citep{ot_watermark}, originally derived using optimal transport.
Indeed, in \cref{theorem:aar_watermarks}, we find that the resulting sampling mechanisms for any $G$ are always pseudo-deterministic but lack a closed form: they require solving $|\Sigma|$ equations that depend on the distribution of $G$.
With an equality constraint ($\varepsilon = 0$), this is exactly the optimal transport problem introduced in~\citet{ot_watermark}.
Hence, setting $\varepsilon \ge 0$ also extends the schemes from~\citet{ot_watermark} with a distortionary version.

\paragraph{SynthID}
Another prominent watermarking scheme is SynthID~\citep{synthid}, which uses $m$-layer tournament sampling as a watermark sampling mechanism.

To capture the base $m=1$ layer tournament, we use i.i.d. Bernoulli scores $G$ and a $\chi^2$ distance as a constraint,
\begin{equation} \label{eq:chi2_constraint}
    D_{\chi^2}(q(g), \mathbb{E}_G[q(G)],p) := \sum_{u \in \Sigma} \frac{\left(q(g)_u - p_u\right)^2}{p_u}.
\end{equation}
Given $\varepsilon$ small enough, \cref{theorem:synthid_text} shows that the solution is given by
\begin{equation} \label{eq:1_layer_tournament}
    q_{\text{1 layer}}(g) := p(1 + (g - p \cdot g)).
\end{equation}
This corresponds to the vectorized 1-layer tournament sampling formula derived in~\citet{synthid}.
To capture the $m$-layer tournament, we show in \cref{theorem:synthid_text} that it directly corresponds to chaining the optimization problem $m$ times, each time using the $\chi^2$ constraint with respect to the previous solution.
This formulation intuitively offers a new perspective on why increasing the number of tournament layers $m$ improves the watermark strength: at each layer of the tournament, we reapply the $\chi^2$ watermark.

\paragraph{$\chi^2$ Watermark}
With \cref{eq:chi2_constraint}, we can also extend the 1-layer tournament beyond the SynthID formulation.
In particular, with arbitrary $\varepsilon \ge 0$ and given any $g \in \mathbb{G}$, we find in \cref{theorem:chi2_watermarks} that the sampling mechanism corresponding to the $\chi^2$ constraint is given by
\begin{equation} \label{eq:chi2_eq_wm}
    \forall u \in \Sigma, q^*(g)_u = p_u [1 + \delta(g_u + \mu)]_+
\end{equation}
where $\mu \in \mathbb{R}$ is chosen to satisfy the simplex constraint, $\delta \ge 0$, and $[\,\cdot\,]_+$ is the positive part function.
This offers a new approach to distortionary 1-layer tournament sampling.
In SynthID, greater power can be obtained at the cost of increased distortion by increasing the number of candidates per tournament.
However, this only allows power to be increased in discrete steps.
By contrast, in \cref{eq:chi2_eq_wm}, increasing $\delta$ increases power while trading off distortion in a continuous manner.
We evaluate this scheme using binomial scores in \cref{sec:evaluation}.

\input{figures/metric_eval/main_metric_eval.tex}

\subsection{Designing Optimal Watermarking Schemes Given Specific Constraints}
\label{sec:method_practice:new}

Next we show how to design optimal watermarking sampling mechanisms using our proposed framework with respect to specific constraints.
As long as computing the solution to \cref{eq:constrained_optimization_problem} is efficient, one can choose any constraints on the simplex as $D$, \eg the total variation distance or the Hellinger distance.
To highlight the practicality of our framework, we propose two new schemes, one hard- and one soft-constrained, where the constraints specifically target the watermarked text quality.
We find in \cref{sec:evaluation:detectability} that those new schemes indeed have a superior detectability-quality trade-off than prior works.

\paragraph{Perplexity As a Distance}
The per-token perplexity (PPL) is defined as the exponential of the information entropy.
It measures how surprising it is that a given token has been sampled.
The perplexity of a sequence of tokens is the average of the individual token perplexities.
In the broader NLP literature, and in the watermarking field in particular, it has been used as a proxy to measure the quality impact of a watermark~\citep{kgw1,synthid}.
Notably, our framework allows us to directly use the perplexity (measured according to the original model probability $p$) as a constraint when designing watermarking sampling mechanisms.
Specifically, we introduce both a hard-constraint and a soft-constraint variant, for all $g \in \mathbb{G}$,
\begin{align}
    & D_{\text{Hard PPL}} := (p-q(g))\cdot \log p \label{eq:hard_ppl_constraint} \\
    & D_{\text{Soft PPL}} := (p-\mathbb{E}_G[q(G)])\cdot \log p.
\end{align}

\paragraph{Hard PPL Constraint}
Given $\varepsilon \ge 0$, finding the watermark sampling mechanism using the hard PPL constraint is equivalent to solving the following optimization problem
\begin{equation} \label{eq:ppl_hard_constraint_op}
    \begin{aligned}
    q^* =  &\argmax_{q: \mathbb{G} \rightarrow \Delta(\Sigma)} \quad \mathbb{E}_G [G \cdot q(G)] \\
    & \text{s.t.} \, \forall g \in \mathbb{G}, D_{\text{Hard PPL}}\left(q(g), \mathbb{E}_{G}[q(G)],p\right) \le \varepsilon.
    \end{aligned}   
\end{equation}
This is a linear programming problem that, while no closed-form solution exists, can be solved efficiently using standard solvers.
In particular, given $g \in \mathbb{G}$, the optimal distribution $q^*(g)$ is almost deterministic: the optimal solution is the barycenter of two Dirac distributions and hence of support at most $2$.
This result is intuitive: if one wants to minimize the PPL of a sequence of tokens, one simply has to do greedy sampling.
An intuitive way to solve the optimization problem in \cref{eq:ppl_hard_constraint_op} is to order the tokens according to their probability $p$.
There exists a rank $r > 0$ such that deterministically sampling any token below this rank satisfies the constraint.
If the maximum of $G$ is reached in this subset, we sample this token.
Otherwise, find the best combination of a token below and a token above the threshold that maximizes the power.
We evaluate this scheme in \cref{sec:evaluation}, and find that it is a powerful scheme with high diversity. 

\paragraph{Soft PPL Constraint}
With the soft PPL constraint, we show in \cref{theorem:distortion_free_argmax} that, because we are only constraining $\mathbb{E}_{G}[q(G)]$, the solution $q^*(G)$ is a Dirac distribution. 
Specifically, we show in \cref{theorem:soft_ppl_watermark} that we have
\begin{equation}
    \exists \lambda \! \in \mathbb{R}, \forall u \! \in \Sigma, q^*\!(g)_u =\! \argmax_{v\in\Sigma}(g_{v}\!+\!\lambda(\varepsilon) \log p_{v}),
\end{equation}
where $\lambda(\varepsilon)$ is the unique solution to $\mathbb{E}_{G}[q^*(G) \cdot \log p]= p \cdot \log p - \varepsilon$.
Interestingly, when $G$ is composed of i.i.d. Gumbel random variables and $\varepsilon = 0$, then $\lambda(0) = 1$, which is exactly the AAR/KTH sampling mechanism.
While in general, the feasible set of distortion-free constrained watermarks is contained in the set of soft PPL constrained watermarks, with i.i.d. Gumbel random variables the solution coincides.
For other distributions, however, there is no closed-form solution for $\lambda(\varepsilon)$, requiring a root-finding algorithm to approximate $\lambda(\varepsilon)$.
While this makes the watermark itself more expensive, our framework remains consistent: Our evaluation in \cref{sec:evaluation} shows that this scheme outperforms all other schemes in the detectability-quality trade-off.

%% file: tables/table_overview.tex
\begin{table*}[t]
    \centering
    \caption{\textbf{Overview of Our Framework} We summarize which constraints correspond to which watermarks, including prior works and the novel watermarks introduced in~\cref{sec:method_practice}.
    We highlight in \colorbox{green!10}{light green} the novel schemes introduced in this work.
    We recall that $p$ is the model next-token probability distribution, $G$ the random variable representing the pseudorandom token scores, $g$ the sampled token scores (via the hashing mechanism), and $q$ the watermarked probability distribution.
    For SynthID, we chain $m$ constrained optimization problems corresponding to the $m$ tournament layers and let $q^{(0)} = p$.
    Given that, in this work, we do not extend the SynthID algorithm directly, we use the same algorithm as the one from~\citep{synthid}.
    All the algorithms are presented in \cref{app:algorithms}.}
    \label{tab:overview}
  
    \renewcommand{\arraystretch}{1.2}
    \newcommand{\skiplen}{0.000001\linewidth} 
    \newcommand{\rlen}{0.01\linewidth} 
    \resizebox{\linewidth}{!}{%
    \begingroup 
    \setlength{\tabcolsep}{5pt} %
  \begin{tabular}{lcccccc}
  \toprule
   \textbf{Scheme}  & Red-Green & AAR/KTH & \cellcolor{green!10}$\chi^2$ & SynthID & \cellcolor{green!10}PPL (Hard) & \cellcolor{green!10}PPL (Soft) \\
  \midrule
   \textbf{Constraint} & $\text{KL}(q(g) || p)$ & $\text{KL}(\mathbb{E}_G[q(G)] || p)$ & $\chi^2(q(g),p)$  & $\{\chi^2(q^{(i)}(g),q^{(i-1)}(g))\}_{i=1,\dots,m}$ & $(p-q(g))\cdot \log p$ &  $(p-\mathbb{E}_G[q(G)])\cdot \log p$ \\
   \textbf{$G$ Distribution} & Bernoulli & Gumbel & Binomial & Binomial & Binomial & Binomial  \\
   \midrule
   \textbf{Algorithm} & \cref{alg:red_green_watermark} & \cref{alg:aar_watermark} & \cref{alg:chi2_watermark} & \citet{synthid} & \cref{alg:ppl_hard_watermark} & \cref{alg:ppl_soft_watermark} \\
  \bottomrule
  \end{tabular}
    \endgroup
    }
  \end{table*}

%% file: figures/metric_eval/main_metric_eval.tex
\begin{figure*}[t]
    \centering
    \includegraphics[width=0.33\textwidth]{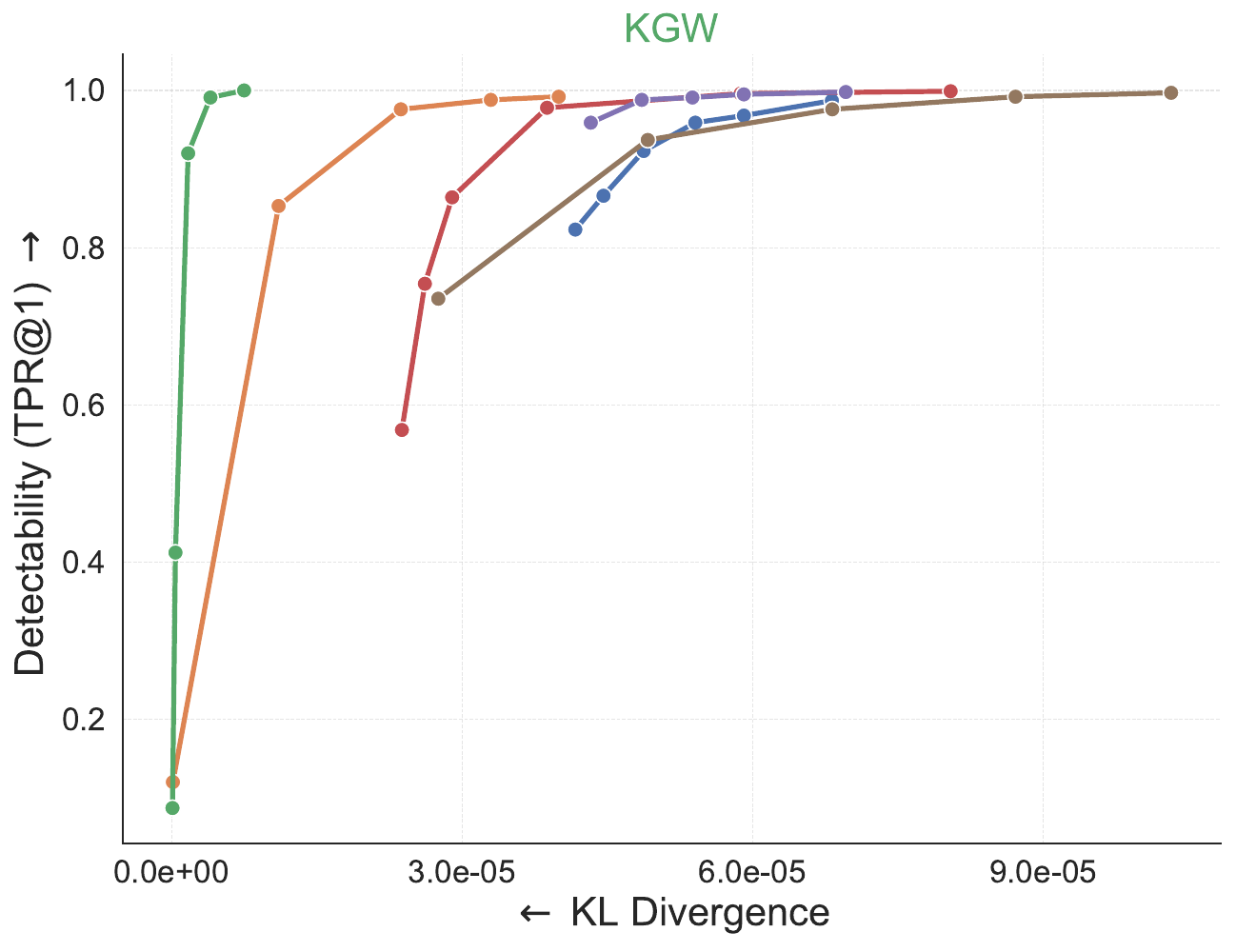}
    \includegraphics[width=0.33\textwidth]{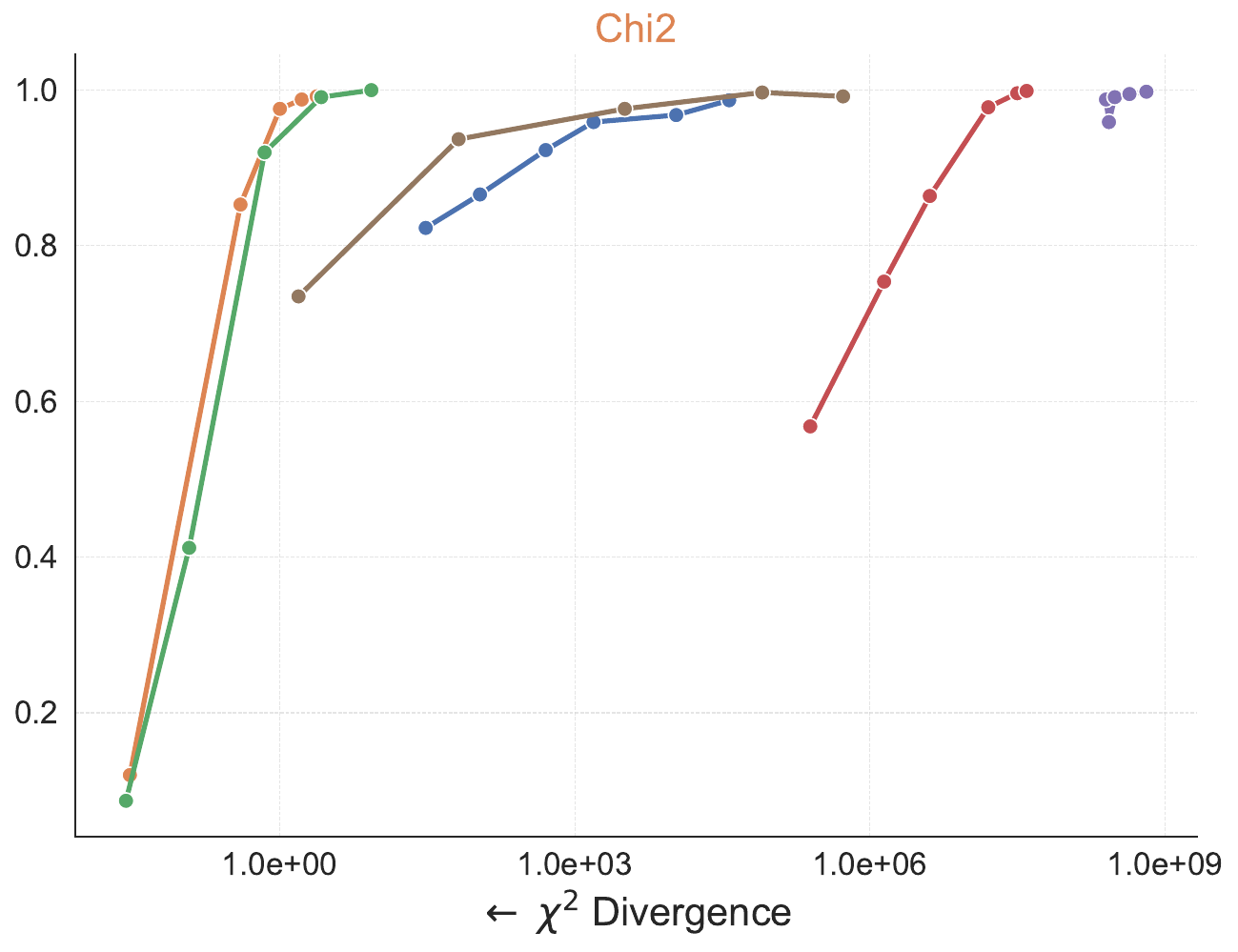}
    \includegraphics[width=0.33\textwidth]{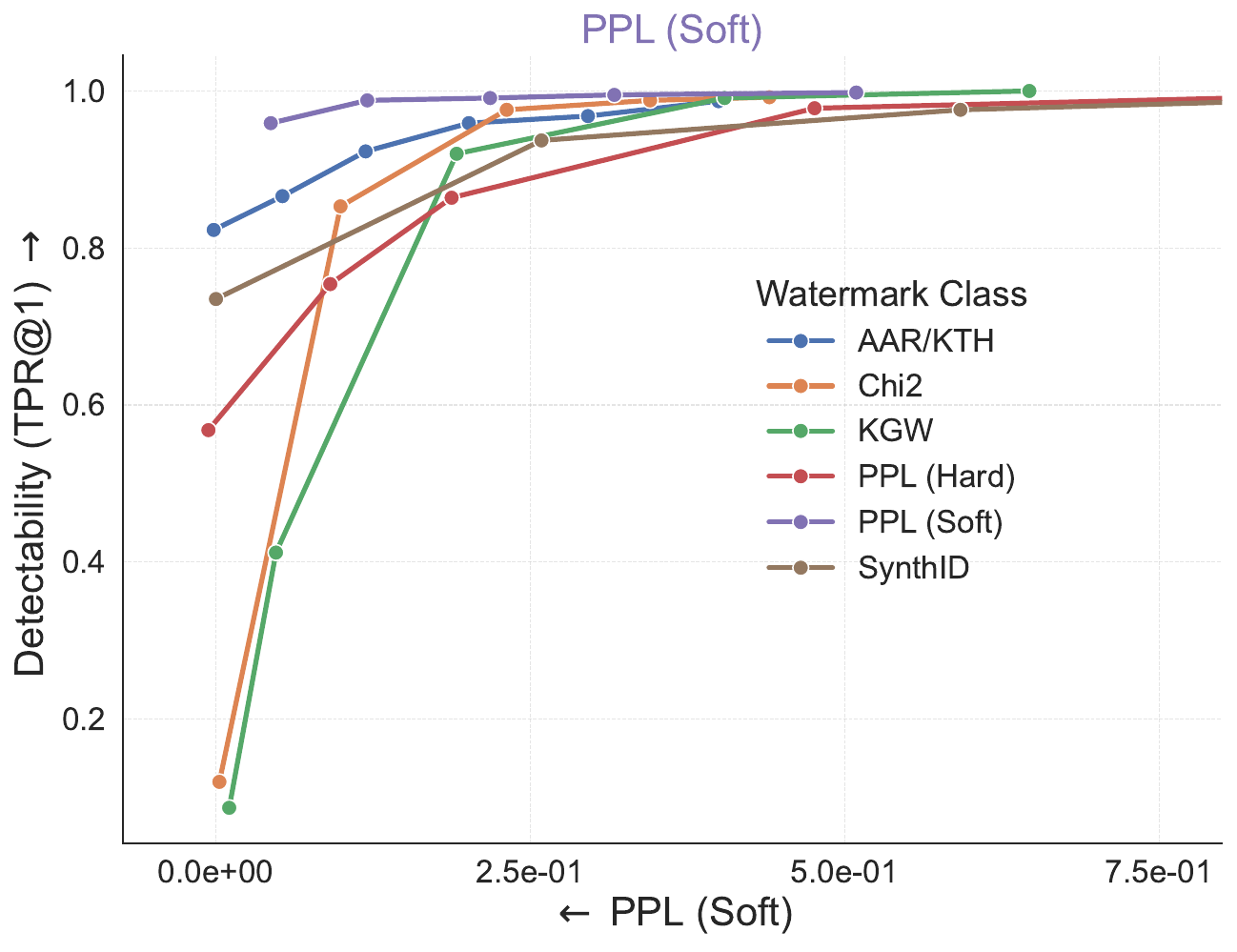}
    \caption{\textbf{Comparison of the Detectability–Constraint Trade-off.} 
    We compare the trade-off between watermark detectability (TPR@1) and different constraints (KL divergence (\emph{left}), $\chi^2$  distance (\emph{middle}), and soft PPL (\emph{right})).
    We find that, for each constraint, the corresponding scheme has the best detectability-constraint trade-off.
    Responses are $200$-token-long replies by \llama{} with temperature $0.7$ and $1000$ prompts from ELI5.
    We highlight above each subplot the scheme derived from the plot constraint.
    }
    \label{fig:main_metric_eval}
\end{figure*}

%% file: src/evaluation.tex
\section{Evaluation}
\label{sec:evaluation}

In this section, we experimentally validate our framework.
In \cref{sec:evaluation:metric}, we show that schemes derived from a constraint $D$ are Pareto-optimal with respect to $D$.
For instance, when measuring TPR@1 with respect to the (hard) KL constraint we find that Red-Green watermarks are indeed Pareto-optimal.
In \cref{sec:evaluation:detectability}, we show that both the new schemes introduced in \cref{sec:method_practice} have a state-of-the-art detectability-quality trade-off compared to prior works.~Lastly, in \cref{sec:evaluation:diversity}, we confirm the fundamental quality-diversity-power trade-off highlighted in \cref{sec:method_theory} by introducing a new experimental procedure tailored to measure this trade-off explicitly.

\input{figures/tpr_ppl_plots/tpr_ppl.tex}

\subsection{Optimality with Respect to the Constraint}
\label{sec:evaluation:metric}

Next, we experimentally validate our unified framework.
Specifically, we demonstrate that a watermarking mechanism designed using a given constraint is Pareto-optimal with respect to that constraint.
Overall, our approach enables the straightforward design of high-performing watermarks tailored to a specific metric of interest, such as distortion-freeness or perplexity.

\paragraph{Experimental Setup}
Unless stated otherwise, our experiments use \llama{} with temperature $0.7$.
For each watermark sampling mechanism, we defer the exact hyperparameters (\eg score distribution, detection algorithm, and hashing mechanism) to \cref{app:experimental_details}.
We generate, with each watermark and varying watermark strength parameters, a thousand replies to questions from the ELI5~\citep{eli5} dataset.
Each reply is $200$ tokens long.
For each scheme, we measure the TPR at 1\% FPR, and, for each generated token, the constraint (\eg KL divergence, $\chi^2$ distance) between the original and watermarked next-token probability distributions, which we then average over the sequences.
We detail the exact procedure in \cref{app:additional_eval:metric_optimality_ablation}.
To measure $\mathbb{E}_G[q(G)]$ with soft constraints, we use a Monte Carlo approximation with $128$ samples.

\paragraph{Results}
\cref{fig:main_metric_eval} shows the constraint-detectability trade-off for the KL constraint, the $\chi^2$ constraint, and the soft-PPL.
We defer the other constraints from \cref{sec:evaluation} to \cref{app:additional_eval:metric_optimality_ablation}.
Our results show that, as predicted by our framework, Red-Green watermarks are Pareto-optimal for the KL constraint, the $\chi^2$ scheme is optimal for the $\chi^2$ constraint, and the soft-PPL scheme is optimal for the soft-PPL constraint.
This result shows that our framework allows the principled construction of powerful schemes with respect to a specific constraint, and correctly captures which specific constraints watermarks from prior work optimize.

\subsection{Detectability-Quality Trade-off}
\label{sec:evaluation:detectability}

We next explore the detectability-quality trade-off of the different watermarks presented in \cref{sec:method_practice}.
Importantly, we show that the extensions of existing schemes (distortionary AAR/KTH, distortionary 1-layer tournament) and the newly proposed PPL-constrained schemes outperform prior work.
Notably, our soft PPL scheme achieves a state-of-the-art detectability-quality trade-off, confirming that the soft PPL constraint is a practical proxy for watermark quality.

\paragraph{Experimental Setup}
We use the same experimental setup as in \cref{sec:evaluation:metric}.
As a proxy for quality, we measure the average text perplexity using \qwen{}, and in \cref{app:additional_eval:benchmnarks} we use LLM benchmark accuracy.
Additionally, we show some watermarked text examples in \cref{app:text_example}, and we ablate this experiment on \mistral{} in \cref{app:additional_eval:ministral}.

\paragraph{Results}
In \cref{fig:main_tpr_ppl} we show the detectability-quality trade-off for all watermarks from~\cref{sec:method_practice}, explicitly differentiating between hard-constraint watermarks (left) and soft-constraint watermarks (right).

For hard-constraint watermarks, both the Red-Green watermark and the $\chi^2$ watermark have a similar detectability-quality trade-off.
The hard-PPL watermark exhibits higher detectability in the low-distortion regime.

For soft-constraint watermarks, at $\varepsilon = 0$ (\ie when the log PPL of the watermarked model matches that of the unwatermarked model), the soft-PPL scheme outperforms AAR/KTH, which itself outperforms SynthID.
Indeed, both SynthID and AAR/KTH satisfy the distortion-free constraint.
Yet, under the assumptions from \cref{sec:method_practice:existing} (\ie using i.i.d. Gumbel entries for $G$), AAR/KTH is the optimal sampling mechanism among all distortion-free schemes.
For the soft-PPL scheme, we specifically chose the constraint to maximize the watermark detectability-quality trade-off, which is why it outperforms both schemes.

Importantly, this result is not due to overfitting our quality metric (log PPL): as we find in \cref{app:additional_eval:benchmnarks}, when using LLM benchmark accuracy as a proxy for quality, the soft-PPL scheme still outperforms both AAR/KTH and SynthID.
In the distortionary case, where the PPL of the watermarked model exceeds that of the unwatermarked one, we find that the ordering remains similar.
In particular, this implies that our extension of AAR/KTH enables greater flexibility by trading distortion for better watermark detectability.

Importantly, these results highlight that our framework not only captures popular prior watermarks. It also allows for the principled incorporation of relevant watermarking constraints when designing new and powerful schemes.

\input{figures/diversity_eval/diversity_eval.tex}

\subsection{Impact on Diversity}
\label{sec:evaluation:diversity}

Next, we investigate how different watermarking methods affect output diversity.
Although prior work~\citep{gloaguen2025blackbox} suggested that watermarks, especially distortion-free schemes, may reduce diversity, this effect has not been rigorously evaluated.
We propose a new experimental protocol to quantify the impact of watermarking on LLM output diversity.
Our results show that soft-constrained watermarks substantially reduce diversity, whereas hard-constrained watermarks are able to retain it to a larger extent.

\paragraph{Measuring Watermark Diversity}
As in~\cref{sec:evaluation:detectability}, we use \llama{} with temperature $0.7$.
The effect of a watermark on diversity, as noted in~\citet{gloaguen2025blackbox}, is to reduce the number of unique outputs given a fixed prompt.
This means that, in order to measure the watermark's impact on diversity, we need to sample the model multiple times using the same prompt.
If the watermark negatively impacts diversity, we expect a higher number of repetitions within the generated answer sequences.
In the limit, with AAR/KTH, given a fixed prompt the sampling is deterministic and all the generated sentences are exactly the same: the watermark exhibits low diversity.

We therefore propose the following experimental setup.
Given a dataset of prompts $\{x_i\}_i$, for each prompt, we sample the model $N$ times, computing the self-BLEU score~\citep{self_bleu} for that prompt.
Self-BLEU scores are between $0$ and $1$ and measure the n-gram repetition between texts.
In our setting, a high self-BLEU score means that, for a given prompt, the multiple generations have very similar n-grams.
To smooth out the effect of the prompt itself (some prompts might naturally lead to low-diversity text), we then average all self-BLEU scores.
In practice, we use $100$ prompts from the ELI5 dataset and sample the model $N \coloneq 100$ times per prompt, yielding a total of $10{,}000$ samples.
For each sample, we generate $200$ tokens, using the trigram self-BLEU score with uniform weights.

\vspace{-2mm}
\paragraph{Watermarks Diversity-Quality Trade-off}
\cref{fig:main_diversity_eval} shows the impact on diversity and quality of different watermarks given a fixed TPR@1.
In particular, as predicted by our framework, we find that there is a negative correlation between diversity and quality: given a fixed power, some schemes prioritize diversity over quality and vice versa.
In particular, soft-constrained schemes (AAR/KTH and PPL (Soft)) have particularly low diversity (\ie high self-BLEU scores) but also the highest quality.
This confirms the theoretical derivation from~\cref{sec:method_theory}: at a fixed power, there is a fundamental trade-off between watermark quality and diversity.
More generally, this suggests that to measure the impact of distortion-free watermarks it is also necessary to evaluate the watermark behavior under repetition of the same prompt (a similar finding to~\citet{wu2025analyzing}), and not only assume that distortion-free implies no actual distortion.

%% file: figures/tpr_ppl_plots/tpr_ppl.tex
\begin{figure*}[t]
    \centering
    \includegraphics[width=\textwidth]{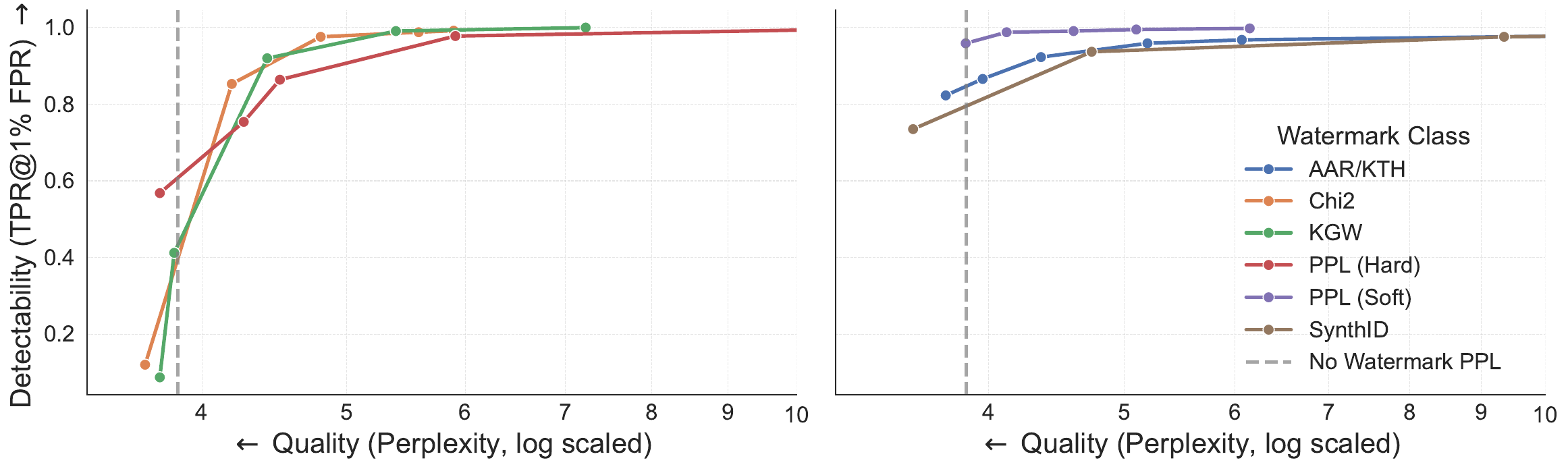}
    \caption{\textbf{Comparison of the Detectability–Quality Trade-off.} 
    We compare the trade-off between watermark detectability (TPR@1) and text quality (log PPL) for different constraint instantiations and for different $\varepsilon$.
    The \emph{left} figure shows the hard constraints and the \emph{right} one shows the soft constraints.
    The dashed line corresponds to the log PPL of the unwatermarked replies.
    Responses are 200-token-long replies by \llama{} with temperature $0.7$ and $1000$ prompts from ELI5.} 
    \label{fig:main_tpr_ppl}
\end{figure*}

%% file: figures/diversity_eval/diversity_eval.tex
\begin{figure}[t]
    \centering
    \includegraphics[width=0.48\textwidth]{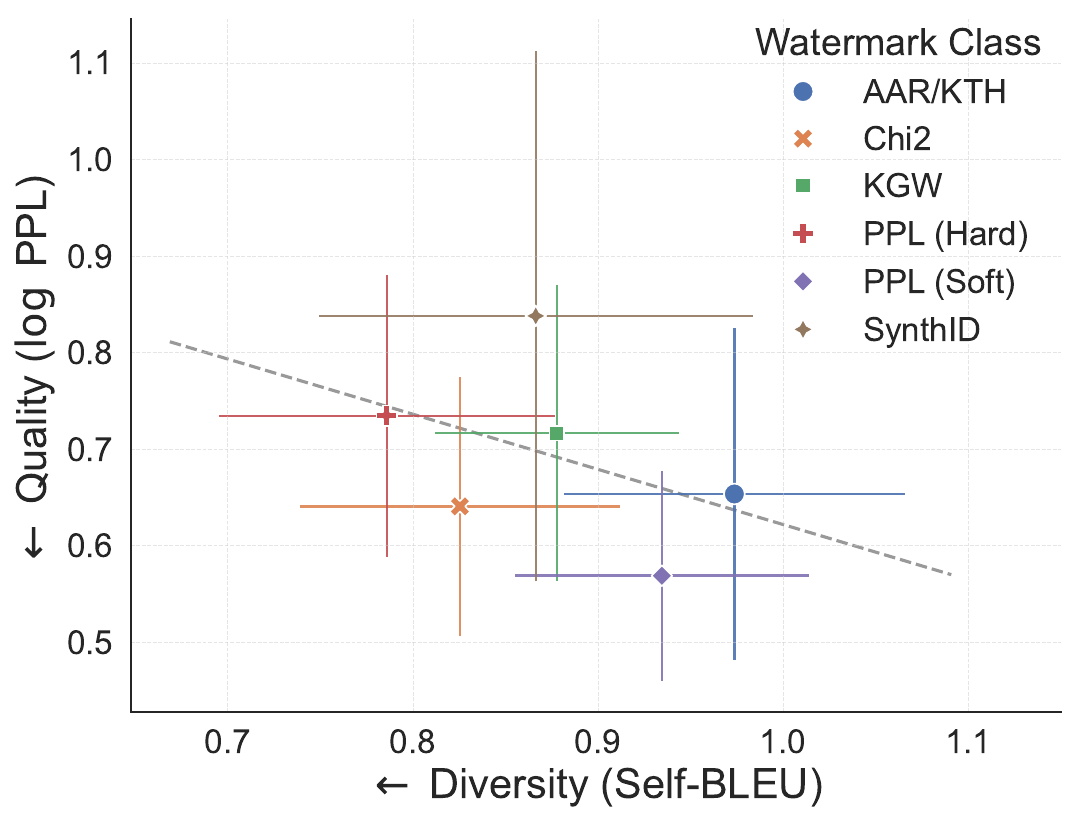}
    \caption{\textbf{Diversity–Quality Trade-off at Fixed TPR.}
        We compare the trade-off between watermark impact on diversity (Self-BLEU) and quality ($\log$ PPL) given a fixed TPR@1 of $0.95$.
        A high self-BLEU score corresponds to a low output diversity.
        The error bars correspond to the standard deviation and the dotted gray line shows the correlation.
        The PPL and TPR@1 are measured over $1,000$ prompts from ELI5 (\cref{sec:evaluation:detectability}), while the diversity is measured with the procedure described in~\cref{sec:evaluation:diversity}.
        We find that, given a fixed power, diversity and quality are negatively correlated.
        }
    \label{fig:main_diversity_eval}
\end{figure}

%% file: src/conclusion.tex
\section{Conclusion and Limitations}
\label{sec:conclusion}

Our work shows that most existing watermarks can be seen as instances of a unified constrained optimization problem.
As confirmed by our evaluation, our framework makes three contributions:
(i) it provides a principled way to interpret prior watermarks,
(ii) it allows us to practically instantiate powerful new schemes tailored to specific watermarking requirements,
and lastly (iii) it provides a novel perspective on the quality-diversity-power trade-off in watermarks.

\vspace{-2mm}
\paragraph{Limitations}
Our work makes two key assumptions when constructing the optimization problem.
First, we restrict the optimization to the token level rather than the sequence level.
Second, we deliberately optimize only the watermark sampling mechanism and do not jointly optimize the detector.
While these assumptions make the optimization tractable and enable us to capture and extend a wide range of watermarking schemes, there may exist jointly optimized approaches that achieve better overall watermarking performance.
Lastly, we note that there exists a range of watermarking algorithms (such as fully cryptographic watermarks) that are not captured by our framework.

%% file: src/impact_statement.tex
\section*{Impact Statement}
This paper presents a unified framework for LLM watermarking, which allows us to better understand existing schemes and design new ones.
Better-understood and better-tuned watermarks could support provenance, auditing, and mitigation of large-scale AI-generated text misuse.
Nonetheless, while LLM watermarks in general, and our framework in particular, are designed to mitigate false positives, they should not be used as sole evidence in major decisions.

%% file: appendix.tex
\input{src/appendix/experimental_details.tex}
\input{src/appendix/algorithms.tex}

\input{src/appendix/additional_evaluation.tex}
\input{src/appendix/proofs.tex}
\input{src/appendix/examples.tex}

%% file: src/appendix/experimental_details.tex
\section{Experimental Details}
\label{app:experimental_details}

In this section, we detail the hyperparameters used for our experiments in \cref{sec:evaluation}.

\paragraph{Watermark Hyperparameters}
For Red-Green watermarks, we used as a score distribution a Bernoulli distribution with parameter $\gamma = 0.5$.
For AAR/KTH, we used a Gumbel distribution with parameters $(0,1)$.
For SynthID, we used a $30$-layer tournament with Bernoulli $g$-values of parameter $0.5$.
Lastly, for the $\chi^2$ watermark, the hard PPL watermark, and the soft PPL watermark, we used a binomial distribution with parameters $(30,0.5)$.
For the hashing mechanism, we use for all schemes the SumHash hashing scheme from~\citet{kgw1} with the four previous tokens as the context.
For detection, for Red-Green watermarks we used a binomial test, for AAR/KTH a Kolmogorov–Smirnov test between the observed score distribution and the Gumbel CDF, for SynthID we used a binomial test, and lastly for $\chi^2$ watermarks and both PPL watermarks, we also used a binomial test.
For all tests, following the guidelines from~\citet{three_bricks}, we removed all (context, token) duplicates prior to the statistical test.

\paragraph{Sampling Hyperparameters}
For generating replies, we used the instruct version of \llama{} with a temperature of $0.7$ and $16$-bit quantization.
We also used top-$k$ with $k=50$ and constrained the model to generate $200$-token-long replies.
For measuring perplexity, we use the instruct version of \qwen{} with a temperature of $1.0$ and $8$-bit quantization.

%% file: src/appendix/algorithms.tex
\section{Algorithms}
\label{app:algorithms}

In this section, we examine the different sampling mechanisms for the watermarks introduced in \cref{sec:method_practice}.
We recall that, as explained in \cref{sec:method_theory}, we use $\varepsilon$ as a scheme parameter when it corresponds to the constrained formulation (\cref{eq:constrained_optimization_problem}) and $\delta$ when it corresponds to the penalized formulation (\cref{eq:penalized_optimization_problem}).

\subsection{Red-Green watermarks}

\input{algorithms/red_green.tex}

\subsection{AAR/KTH watermarks}

\input{algorithms/aar.tex}

\subsection{Chi2 watermark}

\input{algorithms/chi_two.tex}

\subsection{PPL watermarks}

\input{algorithms/ppl.tex}

%% file: algorithms/red_green.tex
\begin{algorithm}[H]
    \caption{Red-Green Watermark Sampling Mechanism}
    \label{alg:red_green_watermark}
    \begin{algorithmic}[1]
    \INPUT Next-token probability $p$, watermark strength $\delta$, current scores $g \in \mathbb{R}^{|\Sigma|}$.
    \STATE $q \gets p \exp(\delta g)$
    \STATE $q \gets \frac{1}{||q||_1} q$ \COMMENT{ Normalization}
    \RETURN $q$
    \end{algorithmic}
\end{algorithm}

%% file: algorithms/aar.tex
\begin{algorithm}[H]
    \caption{AAR/KTH Watermark Sampling Mechanism}
    \label{alg:aar_watermark}
    \begin{algorithmic}[1]
    \INPUT Next-token probability $p$, watermark strength $\delta$, current scores $g \in \mathbb{R}^{|\Sigma|}$.
    \STATE $i \gets \argmax(g + \frac{1}{1 + \delta}\log p)$
    \STATE $q_i \gets 1$
    \STATE $\forall u \neq i, \, q_u \gets 0$
    \RETURN $q$
    \end{algorithmic}
\end{algorithm}

%% file: algorithms/chi_two.tex
\begin{algorithm}[H]
    \caption{Compute Optimal Shift $\mu$}
    \label{alg:find_mu}
    \begin{algorithmic}[1]
    \INPUT Next-token probability $p$, watermark strength $\delta$, current scores $g \in \mathbb{R}^{|\Sigma|}$.
    \STATE Sort indices $\pi$ such that $g_{\pi_1} \ge g_{\pi_2} \ge \dots \ge g_{\pi_{|\Sigma|}}$
    \STATE Initialize active support $\mathcal{S} \gets \emptyset$
    \FOR{$i = 1$ \textbf{to} $|\Sigma|$}
        \STATE $\mathcal{S} \gets \mathcal{S} \cup \{\pi_i\}$ 
        \STATE $\mu \gets \frac{1 - \sum_{j \in \mathcal{S}} p_j (1 + \delta g_j)}{\delta \sum_{j \in \mathcal{S}} p_j}$ 
        \IF{$\delta \mu \ge -1 - \delta g_{\pi_i}$}
            \RETURN $\mu$
        \ENDIF
    \ENDFOR
    \RETURN $\mu$
    \end{algorithmic}
\end{algorithm}

\begin{algorithm}[H]
    \caption{$\chi^2$ Watermark Sampling Mechanism}
    \label{alg:chi2_watermark}
    \begin{algorithmic}[1]
    \INPUT Next-token probability $p$, watermark strength $\delta$, current scores $g \in \mathbb{R}^{|\Sigma|}$.
    \STATE $\mu \gets \text{Compute\_mu}(p,\delta,g)$
    \STATE $q \gets p [1 + \delta(g + \mu)]_+$
    \RETURN $q$
    \end{algorithmic}
\end{algorithm}

%% file: algorithms/ppl.tex
\begin{algorithm}[H]
    \caption{PPL (Hard) Watermark Sampling Mechanism}
    \label{alg:ppl_hard_watermark}
    \begin{algorithmic}[1]
    \INPUT Next-token probability $p$, watermark strength $\varepsilon$, current scores $g \in \mathbb{R}^{|\Sigma|}$.
    \STATE $q \gets \text{Solve\_LP}(\argmax_{x \in \Delta(\Sigma)} x \cdot g, \text{ s.t. } -x \cdot \log p \le -p \cdot \log p + \varepsilon)$ \COMMENT{Use convex-hull solver}
    \RETURN $q$
    \end{algorithmic}
\end{algorithm}

\begin{algorithm}[H]
    \caption{PPL (Soft) Constraint}
    \label{alg:ppl_soft_watermark_objective}
    \begin{algorithmic}[1]
    \REQUIRE Next-token probability $p$, watermark strength $\varepsilon$, current scores $g \in \mathbb{R}^{|\Sigma|}$.
    \INPUT $\beta$
    \STATE $i \gets \argmax (g + \beta \log p)$
    \RETURN $\log p_i - p \cdot \log p + \varepsilon$
    \end{algorithmic}
\end{algorithm}

\begin{algorithm}[H]
    \caption{PPL (Soft) Watermark Sampling Mechanism}
    \label{alg:ppl_soft_watermark}
    \begin{algorithmic}[1]
    \REQUIRE Number of Monte Carlo samples $n$.
    \INPUT Next-token probability $p$, watermark strength $\varepsilon$, current scores $g \in \mathbb{R}^{|\Sigma|}$, distribution $G$.
    \STATE $g_1,\dots,g_n \gets \text{Sample}(G)$ \COMMENT{We sample once at the beginning}
    \STATE $\beta \gets \text{Find\_root}(\sum_{i=1}^{n} \text{PPL\_Constraint}(p,\varepsilon,g_i))$ \COMMENT{We use a bisection solver.}
    \STATE $i \gets \argmax (g + \beta \log p)$
    \STATE $q_i \gets 1$
    \STATE $\forall u \neq i, \, q_u \gets 0$
    \RETURN $q$
    \end{algorithmic}
\end{algorithm}

%% file: src/appendix/additional_evaluation.tex
\section{Additional Evaluation}
\label{app:additional_eval}

In this section, we include additional ablation experiments to complement our main experiments from \cref{sec:evaluation}. 
In \cref{app:additional_eval:metric_optimality_ablation}, we extend \cref{sec:evaluation:metric} to the soft KL constraint and the hard PPL constraint. 
In \cref{app:additional_eval:distrib_ablation}, we ablate the effect of the score distribution $G$ on each of the watermarks from \cref{sec:method_practice}.
In \cref{app:additional_eval:ministral}, we confirm the results from \cref{sec:evaluation:detectability} using \mistral{}.
Lastly, in \cref{app:additional_eval:benchmnarks}, we show the quality–detectability trade-off of each of the watermarks from \cref{sec:method_practice} using LLM benchmark accuracy as a proxy for quality.

\subsection{Ablation on Measuring Optimality With Respect to the Constraint}
\label{app:additional_eval:metric_optimality_ablation}

Here, we extend the experiments from \cref{sec:evaluation:metric} to include the missing constraints, namely the hard PPL constraint from \cref{eq:hard_ppl_constraint} and the soft KL constraint from \cref{eq:aar_constraint}.

\input{figures/appendix/ablation_metric/metric_ablation.tex}
\input{figures/appendix/ablation_distrib/distrib_ablation.tex}

\paragraph{Experimental Setup}
We follow the same experimental setup as in \cref{sec:evaluation:metric}.
In particular, for each scheme, we measure the constraint between the original and watermarked next-token probability distributions, and we average it over the sequences.
This means that, given $n$ sequences of length $L$ and the corresponding score matrices $G^{(1)},\dots,G^{(n)} \in \mathbb{R}^{L \times |\Sigma|}$ and LLM probability distributions $p^{(1)}, \dots, p^{(n)} \in \Delta(\Sigma)^L$, we actually measure
\[
\frac{1}{n} \sum_{i=1}^{n} \frac{1}{L} \sum_{t=1}^L D \Big( q(p^{(i)}_t, G^{(i)}_t), \, \mathbb{E}_{G'} \big[ q(p^{(i)}_t, G') \big], \, p^{(i)}_t \Big) 
\approx 
\mathbb{E}_{p,G} \Bigg[ D \Big( q(p,G), \, \mathbb{E}_{G'} \big[ q(p,G') \big], \, p \Big) \Bigg]
\]
where we explicitly show the dependence of $q$ on $p$.
For the hard PPL metric in particular, because it is linear with respect to $q$, this means that in practice we measure the soft PPL metric.
A second limitation of our comparison is that not all schemes use the same distribution for $G$.
This means that the objective and the detector are not the same; it also means that there is no guarantee that a scheme optimized for a given distribution of $G$ outperforms (in terms of the constraint-quality trade-off) another scheme that is not derived from this particular constraint but uses a different distribution for $G$.

\paragraph{Results}
\cref{fig:app:metric_ablation} illustrates both of these limitations.
For the KL divergence between the expected watermark distribution and the next-token distribution, we see that AAR, despite being derived from this constraint, is not optimal.
This is because, unlike most other schemes, AAR uses the Gumbel distribution for $G$, while other schemes use either the binomial distribution or the Bernoulli distribution.
For the hard PPL scheme, similarly, we find that it does not appear to be optimal in \cref{fig:app:metric_ablation}.
This is because, in practice, we can only measure the expected constraint.
For the PPL constraint, the expected hard PPL constraint corresponds exactly to the soft PPL constraint, hence we find that \cref{fig:app:metric_ablation} (right) and \cref{fig:main_metric_eval} (right) are very similar.

\subsection{Ablation on the Score Distribution}
\label{app:additional_eval:distrib_ablation}

Here, we ablate the score distribution $G$.
For all the schemes from \cref{sec:method_practice} except AAR/KTH and SynthID, there is no fundamental restriction on which distribution to choose for $G$.
Yet, our framework does not indicate which score distribution to choose to maximize power: in the optimization problem from \cref{eq:constrained_optimization_problem}, the distribution is treated as a constant.
This means that we have to find experimentally which distribution is best to maximize the detectability–constraint trade-off for each watermark.

\paragraph{Experimental Setup}
We follow closely the experimental setup from \cref{sec:evaluation:metric}.
Specifically, for each scheme, we measure its respective constraint between the original and watermarked next-token probability distributions, and we average it over the sequence.
We also measure the perplexity (computed with \qwen{}) to evaluate the detectability–quality trade-off.
Each scheme is evaluated with $6$ different distributions for the token scores $G$: Bernoulli with success $0.5$, binomial with parameters $(30,0.5)$, Gumbel with parameters $(0,1)$, lognormal with parameters $(0,1)$, normal with parameters $(0,1)$, and uniform over $[0,1]$.

\paragraph{Results}
\cref{fig:app:distrib_ablation} shows for each scheme the different trade-offs when using various score distributions.
In particular, we find that for hard-constraint watermarks, the choice of the score distribution has almost no influence.
For all the distributions tested, the resulting watermarks had similar trade-offs.
One exception is the $\chi^2$ watermark with a Bernoulli distribution, where we find that increasing the distortion beyond a threshold leads to catastrophic failures of the watermarking scheme.
This is why we only show low distortions for the $\chi^2$ watermark with a Bernoulli distribution in \cref{fig:app:distrib_ablation}.

For soft-constraint watermarks, we evaluate the soft PPL scheme in \cref{fig:app:distrib_ablation}.
We find that, in this case, the choice of the distribution $G$ has a significant influence on the quality–detectability trade-off.
Specifically, the heavier the tail of the distribution, the lower the performance of the scheme.
Interestingly, this is the opposite behavior from what~\citet{ot_watermark} observed when extending the AAR scheme to non-Gumbel distributions.
This suggests that there is no single best distribution that works for all watermarking schemes.

\cref{fig:app:distrib_ablation} also highlights an interesting result for the soft PPL watermark.
As stated in \cref{theorem:soft_ppl_watermark}, \cref{alg:ppl_soft_watermark} only holds for continuous distributions (otherwise, to enforce the constraint in addition to $\beta$ one must find the tie-breaking rule).
Hence, we see that using \cref{alg:ppl_soft_watermark} with a Bernoulli distribution leads to a more distortionary watermark than expected.

\subsection{Evaluation on \mistral{}}
\label{app:additional_eval:ministral}

In this section, we reproduce the experiments from \cref{sec:evaluation:detectability} using \mistral{} instead of \llama{} to ensure that the newly proposed watermarks generalize to models of different sizes and from different families.

\input{figures/appendix/ablation_ministral/tpr_ppl_ministral.tex}

\paragraph{Results}
\cref{fig:ministral_tpr_ppl} shows the detectability-quality trade-off for all watermarks from \cref{sec:method_practice}, with hard-constraint watermarks on the left and soft-constraint watermarks on the right.  
We find that, for all watermarks except SynthID, the results are very similar to \cref{fig:main_tpr_ppl}.
In particular, this means that our framework is model-agnostic and can be applied to most LLMs.
Regarding SynthID, we find that at higher distortion levels (not represented in \cref{fig:ministral_tpr_ppl}), model quality significantly deteriorates.

\subsection{Measuring Detectability-Quality Trade-off with LLM Benchmarks}
\label{app:additional_eval:benchmnarks}

In this section, we extend the experiments from \cref{sec:evaluation:detectability} by measuring the detectability-quality trade-off with LLM benchmark accuracy rather than perplexity.
We find that, even when quality is measured through benchmark accuracy, our soft-PPL schemes outperform all prior work.

\input{figures/appendix/ablation_benchmarks/benchmark_ablation.tex}

\paragraph{Experimental Setup}
To measure TPR@1, we follow the experimental setup from \cref{sec:evaluation:detectability}.
For LLM benchmarks, we measure accuracy on three benchmarks with long answers, namely two coding tasks, HumanEval~\citep{humaneval} and MBPP~\citep{mbpp}, and one math task, GSM8K~\citep{gsm8k}.
We use the LM Evaluation Harness framework from~\citet{eval-harness}.
For generating the benchmark answers, we use a temperature of $0.7$ and \llama{}.

\paragraph{Results}
\cref{fig:benchmarks} shows the detectability-quality trade-off, using LLM benchmark accuracy as a proxy for quality, for all watermarks from \cref{sec:method_practice}, with hard-constraint watermarks on the left and soft-constraint watermarks on the right.  
We draw a similar conclusion as in \cref{sec:evaluation:detectability}: both the scheme extensions derived from our framework as well as the novel schemes outperform prior work.
In particular, the soft-PPL scheme maximizes the detectability-quality trade-off across all schemes evaluated.

%% file: figures/appendix/ablation_metric/metric_ablation.tex
\begin{figure*}[t]
    \centering
    \includegraphics[width=0.48\textwidth]{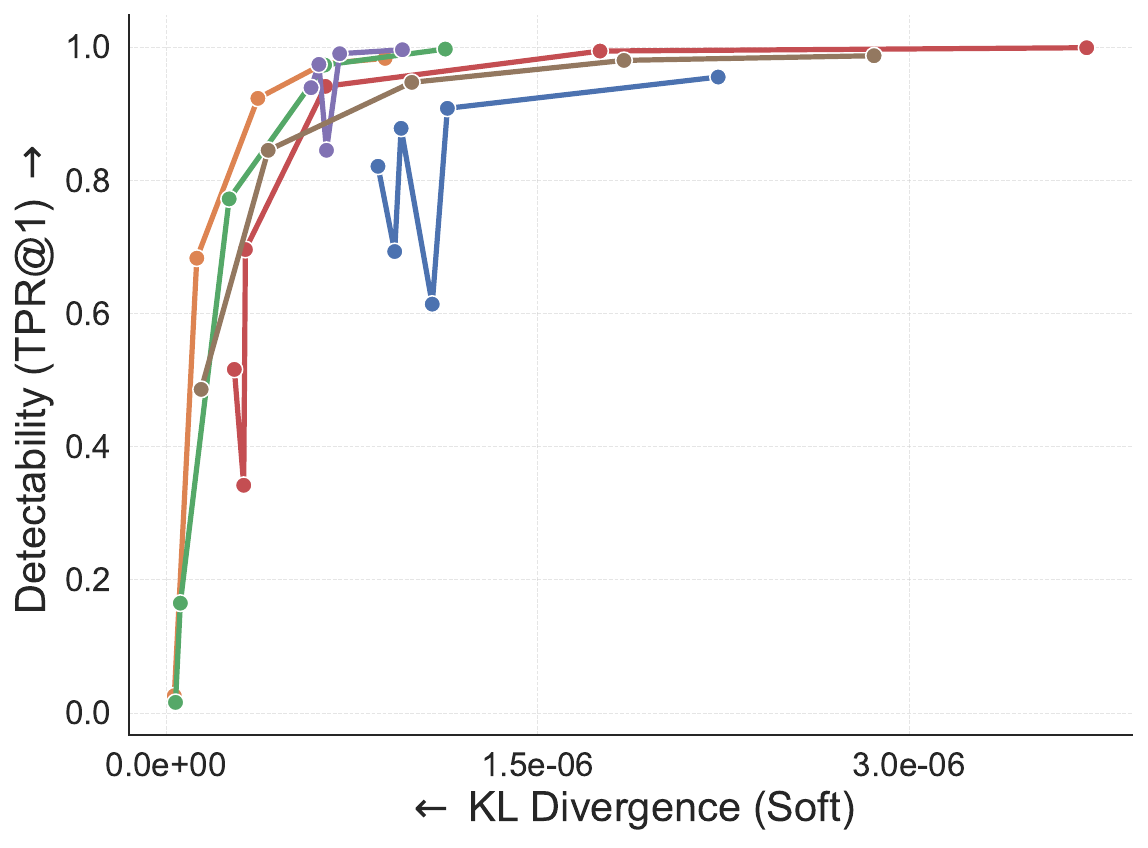}
    \includegraphics[width=0.48\textwidth]{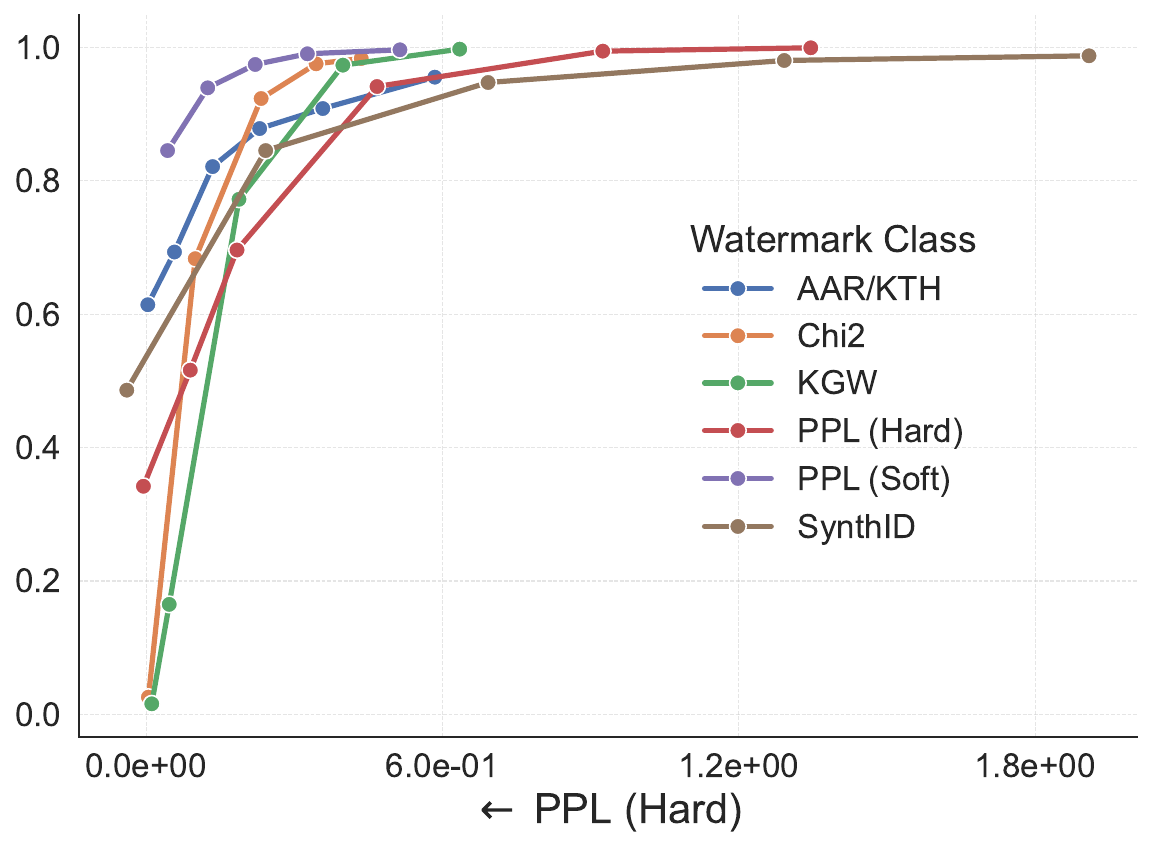}
     \caption{\textbf{Additional Detectability–Constraint Trade-off} 
    We compare the trade-off between watermark detectability (TPR@1) and different constraints (KL divergence between $\mathbb{E}_G[q(G)]$ and $p$ (\emph{left}), and hard PPL (\emph{right})).
    Responses are 200-token-long replies by \llama{} with temperature $0.7$ and $1000$ prompts from ELI5.
    }
    \label{fig:app:metric_ablation}
\end{figure*}

%% file: figures/appendix/ablation_distrib/distrib_ablation.tex
\begin{figure*}[ht!]
    \centering
    \includegraphics[width=\textwidth]{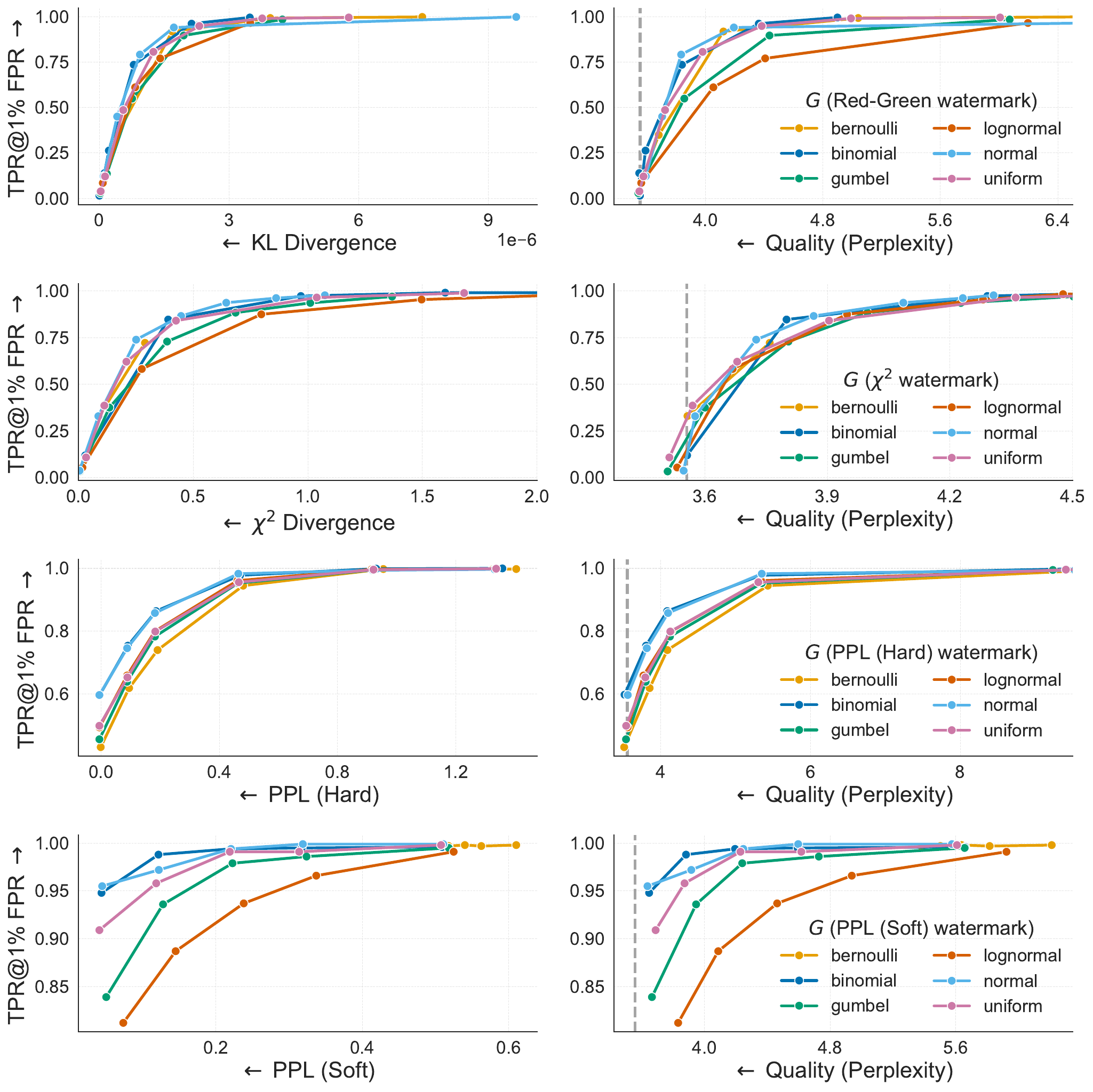}
    \caption{\textbf{Watermarks' Trade-offs for Different Score Distributions.} 
    For each watermark, we compare the trade-off between watermark detectability (TPR@1) and its respective constraint (\emph{left}), and between watermark detectability and perplexity (\emph{right}) for different score distributions $G$.
    The gray dashed line corresponds to the unwatermarked perplexity.
    Responses are 200-token-long replies by \llama{} with temperature $0.7$ and $1000$ prompts from ELI5.
    }
    \label{fig:app:distrib_ablation}
\end{figure*}

%% file: figures/appendix/ablation_ministral/tpr_ppl_ministral.tex
\begin{figure*}[t]
    \centering
    \includegraphics[width=\textwidth]{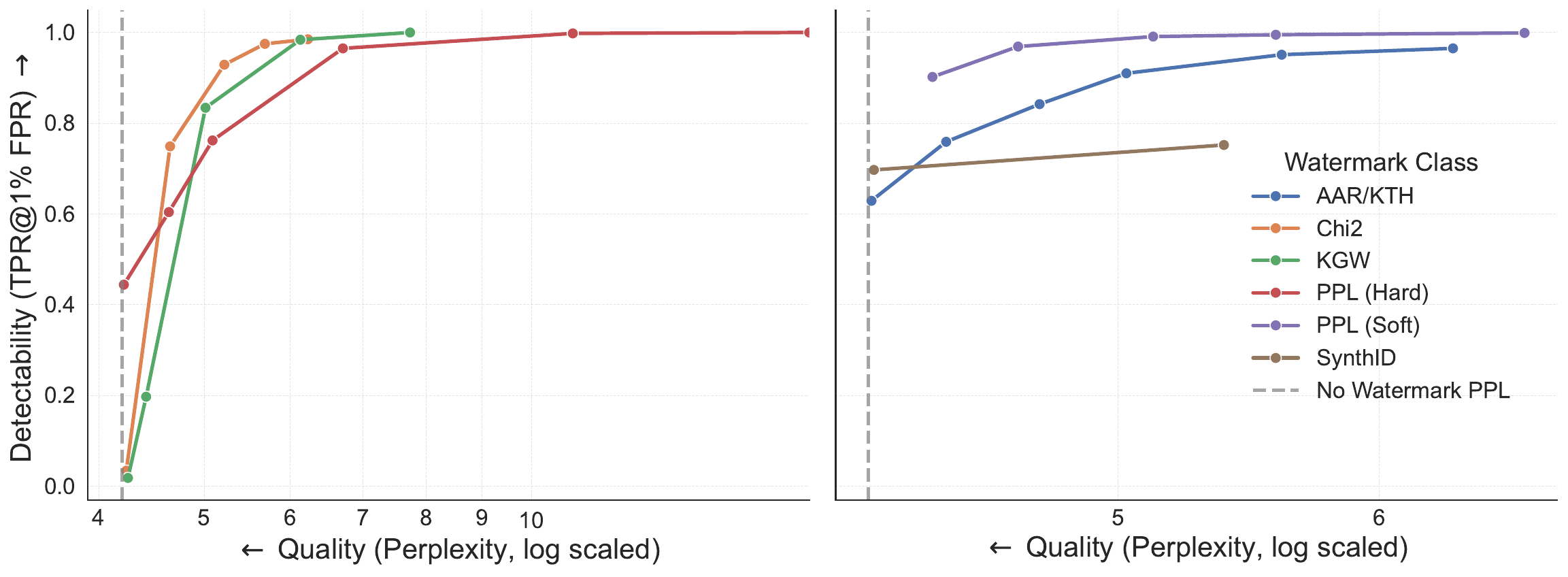}
    \caption{\textbf{Comparison of the Detectability–Quality Trade-off.} 
    We compare the trade-off between watermark detectability (TPR@1) and text quality (log PPL) for different constraint instantiations and for different $\varepsilon$.
    The \emph{left} figure shows the hard constraints and the \emph{right} one shows the soft constraints.
    The dashed line corresponds to the log PPL of the unwatermarked replies.
    Responses are 200-token-long replies by \mistral{} with temperature $0.7$ and $1000$ prompts from ELI5.} 
    \label{fig:ministral_tpr_ppl}
\end{figure*}

%% file: figures/appendix/ablation_benchmarks/benchmark_ablation.tex
\begin{figure*}[t]
    \centering
    \includegraphics[width=\textwidth]{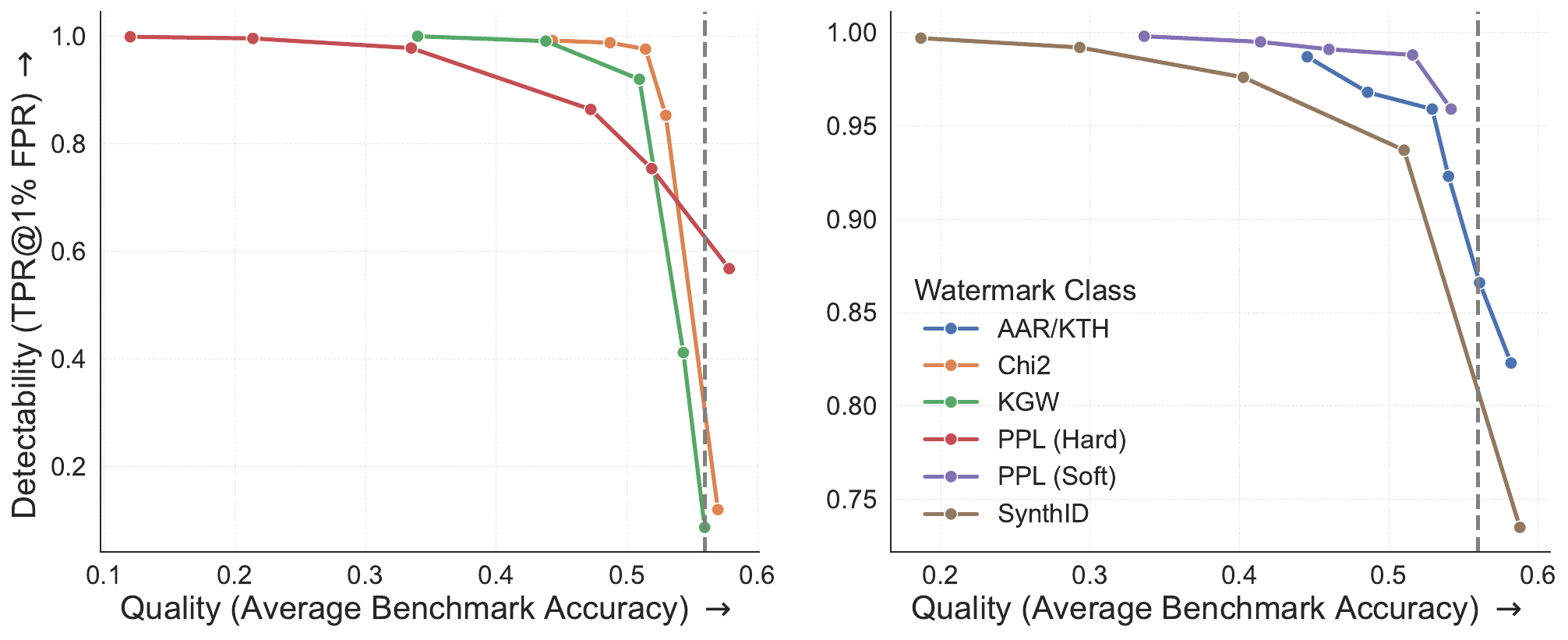}
    \caption{\textbf{Comparison of the Detectability-Quality Trade-off.} 
    We compare the trade-off between watermark detectability (TPR@1) and text quality (LLM benchmark accuracy) for different constraint instantiations and for different $\varepsilon$.
    The \emph{left} figure shows the hard constraints and the \emph{right} one shows the soft constraints.
    The dashed line corresponds to the average accuracy of the unwatermarked model.
    Responses are 200-token-long replies by \mistral{} with temperature $0.7$ and $1000$ prompts from ELI5.
    The benchmark accuracy is averaged over HumanEval, MBPP and GSM8K.} 
    \label{fig:benchmarks}
\end{figure*}

%% file: src/appendix/proofs.tex
\section{Proofs}
\label{app:proofs}

\noindent We follow the notation from the main text: $G$ denotes the (context-dependent) score random vector with support $\mathbb{G}$, $g\in\mathbb{G}$ a realization, and $q(g)\in\Delta(\Sigma)$ the corresponding watermarked distribution.

\subsection{Red-Green Watermarks}
\label{app:proofs:rg_proofs}

Here we formalize and then prove the different results established with respect to Red-Green watermarks.
In particular, we show that KL constraints (\cref{eq:kgw_constraint}) indeed lead to the Red-Green watermark sampling mechanism (\cref{theorem:red_green_watermarks}).
We also extend our framework with the penalized optimization problem \cref{theorem:red_green_watermarks_penalized}.
The penalized optimization problem is exactly the Red-Green watermark from~\citet{kgw1}, where the scheme is parametrized by a strength parameter $\delta$.
With the constrained version, the parameter $\delta$ is only implicitly defined and depends on $\varepsilon$ and $p$.

\begin{theorem}\label{theorem:red_green_watermarks}
Let $p\in\Delta(\Sigma)$ have full support, let $g\in\mathbb{G}$ be non-constant, and let $\varepsilon>0$. Consider
\begin{align*}
q^*(g) \in \argmax_{q(g)\in\Delta(\Sigma)} \quad & g\cdot q(g) \\
\text{subject to}\quad & \mathrm{KL}(q(g)\|p)\le \varepsilon,
\end{align*}
where $\mathrm{KL}(q(g)\|p)=\sum_{u\in\Sigma} q(g)_u\log\frac{q(g)_u}{p_u}$ with the convention $0\log 0 = 0$.

Define, for $\delta\ge 0$,
\begin{equation*}
q_\delta(g)_u\;:=\;\frac{p_u e^{\delta g_u}}{Z_g(\delta)},\qquad Z_g(\delta):=\sum_{v\in\Sigma} p_v e^{\delta g_v},
\end{equation*}
and let
\begin{equation*}
M(g):=\arg\max_{u\in\Sigma} g_u,\qquad \varepsilon_{\max}(g):=-\log\Big(\sum_{u\in M(g)} p_u\Big).
\end{equation*}
Then:
\begin{enumerate}
    \item If $0<\varepsilon<\varepsilon_{\max}(g)$, there exists a \emph{unique} $\delta\in(0,\infty)$ such that $\mathrm{KL}(q_\delta(g)\|p)=\varepsilon$. Moreover, the optimizer is unique and equals $q^*(g)=q_\delta(g)$, i.e.
    \begin{equation*}
    q^*(g)\propto p\,e^{\delta g}.
    \end{equation*}
    \item If $\varepsilon \ge \varepsilon_{\max}(g)$, the distribution
    \begin{equation*}
    q_\infty(g)_u:=\frac{p_u\mathbf{1}\{u\in M(g)\}}{\sum_{v\in M(g)} p_v}
    \end{equation*}
    is optimal.
\end{enumerate}
\end{theorem}

\begin{proof}
Fix $g\in\mathbb{G}$. The problem is a convex optimization problem maximizing a linear function over a convex set (the KL-divergence ball). We introduce a Lagrange multiplier $\lambda \ge 0$ for the inequality constraint $\mathrm{KL}(q(g)\|p) - \varepsilon \le 0$ and a multiplier $\nu \in \mathbb{R}$ for the simplex constraint $\sum_u q(g)_u = 1$. The Lagrangian is:
\begin{equation*}
    \mathcal{L}(q(g), \lambda, \nu) = -\sum_{u\in\Sigma} q(g)_u g_u + \lambda \left( \sum_{u\in\Sigma} q(g)_u \log \frac{q(g)_u}{p_u} - \varepsilon \right) + \nu \left( \sum_{u\in\Sigma} q(g)_u - 1 \right).
\end{equation*}
Note that we minimize $-\mathcal{L}$ to maximize the objective. Taking the derivative with respect to $q(g)_u$ and setting it to zero:
\begin{equation*}
    \frac{\partial \mathcal{L}}{\partial q(g)_u} = -g_u + \lambda \left( \log \frac{q(g)_u}{p_u} + 1 \right) + \nu = 0.
\end{equation*}
Rearranging for $q(g)_u$:
\begin{equation} \label{eq:kkt_form}
    \lambda \log \frac{q(g)_u}{p_u} = g_u - \nu - \lambda \implies q(g)_u = p_u \exp\left( \frac{g_u - \nu - \lambda}{\lambda} \right).
\end{equation}
If $\lambda = 0$, equation \cref{eq:kkt_form} implies $g_u = \nu$ for all $u$ where $q(g)_u > 0$, which contradicts the assumption that $g$ is non-constant. Thus, we must have $\lambda > 0$. Let $\delta \coloneqq 1/\lambda \in (0, \infty)$. Absorbing the normalization constants into a factor $1/Z_g(\delta)$, we obtain the form:
\begin{equation*}
    q_\delta(g)_u = \frac{p_u e^{\delta g_u}}{Z_g(\delta)}, \quad \text{where } Z_g(\delta) = \sum_{v\in\Sigma} p_v e^{\delta g_v}.
\end{equation*}
By the Karush-Kuhn-Tucker (KKT) conditions, specifically complementary slackness, we require $\lambda (\mathrm{KL}(q_\delta(g) \| p) - \varepsilon) = 0$. Since $\lambda > 0$, we must have $\mathrm{KL}(q_\delta(g) \| p) = \varepsilon$, unless the constraint is inactive at the global maximum of the objective over the simplex (Case 2).

Let $\psi_g(\delta) \coloneqq \mathrm{KL}(q_\delta(g) \| p)$. We analyze the behavior of $\psi_g(\delta)$:
\begin{itemize}
    \item As $\delta \to 0$, $q_\delta(g) \to p$, so $\psi_g(0) = 0$.
    \item As $\delta \to \infty$, the mass of $q_\delta(g)$ concentrates on the set $M(g) = \argmax_{u} g_u$. Specifically, $q_\delta(g) \to q_\infty(g)$, where $q_\infty(g)_u = \frac{p_u \mathbf{1}\{u \in M(g)\}}{\sum_{v \in M(g)} p_v}$.
    \item The limit of the KL divergence is:
    \begin{equation*}
        \lim_{\delta \to \infty} \psi_g(\delta) = \sum_{u \in M(g)} \frac{p_u}{\sum_{v \in M(g)} p_v} \log \left( \frac{p_u / \sum_{v \in M(g)} p_v}{p_u} \right) = -\log \left( \sum_{v \in M(g)} p_v \right) = \varepsilon_{\max}(g).
    \end{equation*}
    \item The derivative is $\psi_g'(\delta) = \delta \operatorname{Var}_{U\sim q_\delta(g)}(g_U)$. Since $g$ is non-constant, the variance is positive, so $\psi_g(\delta)$ is strictly increasing.
\end{itemize}

\textbf{Case 1: $0 < \varepsilon < \varepsilon_{\max}(g)$.} \\
Since $\psi_g(0) = 0$, $\lim_{\delta\to\infty} \psi_g(\delta) = \varepsilon_{\max}(g)$, and $\psi_g$ is continuous and strictly increasing, there exists a unique $\delta \in (0, \infty)$ such that $\psi_g(\delta) = \varepsilon$. This $\delta$ satisfies the KKT conditions, making $q_\delta(g)$ the unique optimizer.

\textbf{Case 2: $\varepsilon \ge \varepsilon_{\max}(g)$.} \\
The unconstrained maximum of the linear function $g \cdot q(g)$ over the probability simplex $\Delta(\Sigma)$ occurs at the vertices corresponding to $M(g)$. Among distributions supported on $M(g)$, $q_\infty(g)$ minimizes the KL divergence to $p$ (it is the projection of $p$ onto the face defined by $M(g)$). Since $\mathrm{KL}(q_\infty(g) \| p) = \varepsilon_{\max}(g) \le \varepsilon$, $q_\infty(g)$ is feasible. Since it achieves the global maximum of the objective function $g \cdot q(g)$, it is optimal.
\end{proof}

\begin{theorem}\label{theorem:red_green_watermarks_penalized}
Let $p\in\Delta(\Sigma)$ have full support, let $g\in\mathbb{G}$ be non-constant, and let $\delta>0$. Consider
\begin{align*}
q^*(g) \in \argmax_{q(g)\in\Delta(\Sigma)} \quad & g\cdot q(g) - \frac{1}{\delta} \mathrm{KL}(q(g)\|p),
\end{align*}
We then have, 
$$
    q^*(g)\propto p\,e^{\delta g}.
$$
\end{theorem}

\begin{proof}
Fix $g\in\mathbb{G}$. We formulate the Lagrangian $\mathcal{L}(q(g), \nu)$ for the optimization problem, introducing a multiplier $\nu$ for the normalization constraint $\sum_{u} q(g)_u = 1$:
\begin{equation*}
    \mathcal{L}(q(g), \nu) = \sum_{u \in \Sigma} q(g)_u g_u - \frac{1}{\delta} \sum_{u \in \Sigma} q(g)_u \ln \frac{q(g)_u}{p_u} - \nu \left(\sum_{u \in \Sigma} q(g)_u - 1\right).
\end{equation*}
The first-order optimality condition with respect to $q(g)_u$ is given by setting the derivative to zero:
\begin{equation*}
    \frac{\partial \mathcal{L}}{\partial q(g)_u} = g_u - \frac{1}{\delta} \left( \ln \frac{q(g)_u}{p_u} + 1 \right) - \nu = 0.
\end{equation*}
Rearranging the terms yields:
\begin{equation*}
    \ln \frac{q(g)_u}{p_u} = \delta g_u - (1 + \delta \nu) \implies q(g)_u = p_u e^{\delta g_u} e^{-(1 + \delta \nu)}.
\end{equation*}
Since $e^{-(1 + \delta \nu)}$ is a constant scaling factor independent of $u$ (determined by the constraint $\sum_{u} q(g)_u=1$), we conclude that
\begin{equation*}
    q^*(g)_u \propto p_u e^{\delta g_u}.
\end{equation*}
\end{proof}

\subsection{AAR/KTH Watermarks}

Here we formalize and then prove the different results established with respect to AAR/KTH watermarks.  
We first show the shape of the optimal solution when using soft-KL constraints with any score distribution (\cref{theorem:aar_watermarks}), and then establish the extended watermark mechanism of AAR/KTH (\cref{sec:method_practice:existing}) when using a Gumbel distribution (\cref{theorem:aar_watermarks,theorem:aar_watermarks_penalized}).  
Lastly, we show that for soft-constrained watermarks with a continuous score distribution, the optimal solution is deterministic (\cref{theorem:distortion_free_argmax}), which theoretically demonstrates the lack of diversity in soft-constrained schemes observed in \cref{sec:evaluation:diversity}.

\begin{theorem}\label{theorem:aar_watermarks}
Let $p\in\Delta(\Sigma)$ have full support, let $G$ be a random vector with support $\mathbb{G}\subseteq\mathbb{R}^{|\Sigma|}$ and i.i.d.\ coordinates from a law $\mathcal{L}$, and let $\varepsilon\ge0$. Consider the optimization problem
\begin{align}\label{eq:aar_optimization_problem}
    q^* \in \argmax_{q:\mathbb{G}\to\Delta(\Sigma)} \quad & \mathbb{E}\big[G\cdot q(G)\big] \\
    \text{subject to}\quad & \KL\big(\mathbb{E}[q(G)]\|p\big)\le \varepsilon,
\end{align}
where $\KL(r\|p)=\sum_{u\in\Sigma} r_u\log\frac{r_u}{p_u}$.

Then there exists a vector $\tau\in\mathbb{R}^{|\Sigma|}$ and a measurable tie-breaking function $T: \mathbb{G} \times \Sigma \to [0, 1]$ such that an optimal solution is given by
\begin{equation}\label{eq:aar_argmax_discrete}
    q^*(G)_u \;=\; 
    \begin{cases} 
        1 & \text{if } \{u\} = M(G, \tau), \\
        T_u(G) & \text{if } u \in M(G, \tau) \text{ and } |M(G, \tau)| > 1, \\
        0 & \text{otherwise},
    \end{cases}
\end{equation}
where $M(G, \tau)$ denotes the set of maximizing indices:
\begin{equation*}
    M(G, \tau) = \argmax_{v\in\Sigma}\big(G_v+\tau_v\big).
\end{equation*}
The tie-breaking weights satisfy $\sum_{u \in M(G, \tau)} T_u(G) = 1$ for all $G$, and $T,\tau$ are chosen such that $\KL(\mathbb{E}[q^*(G)]\|p)=\varepsilon$ (unless $\tau=0$ is feasible).

In particular, if $G$ has i.i.d.\ standard Gumbel coordinates, there exists $\alpha\in[0,1]$ such that an optimal solution is given by
\begin{equation*}
\forall u\in\Sigma,\qquad q^*(G)_u \;=\; \mathbbm{1}\Big\{u\in\arg\max_{v\in\Sigma}\big(G_v+\alpha \log p_v\big)\Big\},
\end{equation*}
and $\alpha$ is chosen so that $\mathrm{KL}(\mathbb{E}[q^*(G)]\|p)=\varepsilon$ if feasible, else $\alpha =0$.
\end{theorem}

\begin{proof}
Let $d=|\Sigma|$ and write indices as $u\in\Sigma$.
Set
\begin{equation*}
    r \;=\; \mathbb{E}[q(G)] \in \Delta(\Sigma).
\end{equation*}
We can rewrite~\cref{eq:aar_optimization_problem} as
\begin{align}
\max_{q,r}\quad & \mathbb{E}\big[G\cdot q(G)\big]\\
\text{s.t.}\quad 
& \mathbb{E}[q(G)] = r,\label{eq:aar_primal_rewrite_r}\\
& \mathrm{KL}(r\|p) \le \varepsilon,\label{eq:aar_primal_rewrite_kl}\\
& q(G)\in\Delta(\Sigma)\quad\text{for almost every $G$.} \label{eq:aar_primal_rewrite_simplex}
\end{align}
The feasible set is convex and admits a Slater point, namely $q(G)\equiv p$ (so $r=p$ and $\mathrm{KL}(p\|p)=0<\varepsilon$), hence strong duality and KKT conditions hold.

\paragraph{Shape of an optimizer.}
Introduce $\lambda\in\mathbb{R}^d$ as dual variable for~\cref{eq:aar_primal_rewrite_r} and $\mu\ge 0$ for~\cref{eq:aar_primal_rewrite_kl}. The Lagrangian is
\begin{equation*}
\mathcal{L}(q,r,\lambda,\mu)
=\mathbb{E}\big[G\cdot q(G)\big]
+\lambda\cdot(\mathbb{E}[q(G)]-r)
-\mu\big(\mathrm{KL}(r\|p)-\varepsilon\big),
\end{equation*}
which we rewrite as
\begin{equation}\label{eq:aar_lagrangian_rewrite}
\mathcal{L}(q,r,\lambda,\mu)
= \mathbb{E}\big[(G+\lambda)\cdot q(G)\big]
- \lambda\cdot r
- \mu\big(\mathrm{KL}(r\|p)-\varepsilon\big).
\end{equation}
For fixed $(\lambda,\mu,r)$, maximizing~\cref{eq:aar_lagrangian_rewrite} over $q$ subject to~\cref{eq:aar_primal_rewrite_simplex} separates over $G$:
\begin{equation}\label{eq:aar_inner_sup}
\sup_{q}\mathcal{L}(q,r,\lambda,\mu)
= - \lambda\cdot r - \mu\big(\mathrm{KL}(r\|p)-\varepsilon\big)
+ \mathbb{E}\Big[\sup_{q(G)\in\Delta(\Sigma)} (G+\lambda)\cdot q(G)\Big].
\end{equation}
For any $g\in\mathbb{G}$,
\begin{equation*}
\sup_{q\in\Delta(\Sigma)} (g+\lambda)\cdot q
= \max_{v\in\Sigma} (g_v+\lambda_v),
\end{equation*}
and the supremum is attained by any $q$ supported on the maximizers.

Let $(\lambda^*,\mu^*)$ be an optimal dual pair and define
\begin{equation*}
    M(G,\lambda^*) \;=\; \arg\max_{v\in\Sigma}\big(G_v+\lambda^*_v\big).
\end{equation*}
By strong duality, there exists a primal optimizer $q^*$ that attains the supremum in~\cref{eq:aar_inner_sup} for $(\lambda^*,\mu^*)$. Thus, for almost every $G$, $q^*(G)$ is supported on $M(G,\lambda^*)$.

Choose a measurable tie-breaking function $T:\mathbb{G}\times\Sigma\to[0,1]$ such that
\[
T_u(G)=0\quad\text{if }u\notin M(G,\lambda^*), 
\qquad
\sum_{u\in M(G,\lambda^*)} T_u(G)=1,
\]
and, whenever $M(G,\lambda^*)=\{u\}$ is a singleton, we set $T_u(G)=1$. Then, with $\tau=\lambda^*$,
\begin{equation*}
    q^*(G)_u \;=\;
    \begin{cases}
        1 & \text{if } \{u\} = M(G,\tau),\\
        T_u(G) & \text{if } u\in M(G,\tau)\text{ and }|M(G,\tau)|>1,\\
        0 & \text{otherwise},
    \end{cases}
\end{equation*}
which is exactly the form in~\cref{eq:aar_argmax_discrete}.

Let $q^0(G)_u \propto \mathbbm{1}\Big\{u\in\arg\max_{v\in\Sigma} G_v\Big\}$ a solution of the unconstrained problem, and $r^0=\mathbb{E}[q^0(G)]$.
If $\KL(r^0\|p)\le\varepsilon$, then $q^0$ is feasible for~\cref{eq:aar_optimization_problem} and, being optimal for the unconstrained problem, is also optimal for the constrained one; we may take $\tau=0$.
Otherwise $\KL(r^0\|p)>\varepsilon$, so the KL constraint is active at any optimum by KKT/complementary slackness. Thus there exists an optimal pair $(q^*,r^*)$ with $\KL(r^*\|p)=\varepsilon$, and by the previous paragraph $q^*$ has the form~\cref{eq:aar_argmax_discrete} for some $\tau$ and measurable tie-breaking rule $T$. This proves the first part of the theorem.

\medskip

\paragraph{Gumbel case.}
Assume now that $G$ has i.i.d.\ standard Gumbel coordinates. For any $\tau\in\mathbb{R}^d$, define
\begin{equation}\label{eq:aar_q_tau}
q^\tau(G)_u \;=\; \mathbbm{1}\Big\{u=\arg\max_{v\in\Sigma}\big(G_v+\tau_v\big)\Big\},\qquad
r(\tau)_u \;=\; \mathbb{E}[q^\tau(G)_u].
\end{equation}
The Gumbel-max trick gives
\begin{equation}\label{eq:aar_softmax}
r(\tau)_u \;=\; \frac{e^{\tau_u}}{\sum_{v\in\Sigma} e^{\tau_v}},\qquad u\in\Sigma.
\end{equation}

Fix $\tau$ and write $Z_v=\tau_v+G_v$. For each realization of $G$,
\begin{equation*}
\max_{v} Z_v
= \sum_{u\in\Sigma} Z_u\, q^\tau(G)_u
= G\cdot q^\tau(G) + \tau\cdot q^\tau(G).
\end{equation*}
Taking expectations and denoting $r=r(\tau)$,
\begin{equation}\label{eq:aar_Obj_max_relation}
\mathbb{E}\big[G\cdot q^\tau(G)\big]
= \mathbb{E}\big[\max_{v} Z_v\big] - \tau\cdot r.
\end{equation}
From the computation of the distribution of $Z^*=\max_v Z_v$,
\begin{equation}\label{eq:aar_max_expectation}
\mathbb{E}\big[\max_{v} Z_v\big]
= \gamma + \log\Big(\sum_{v\in\Sigma} e^{\tau_v}\Big),
\end{equation}
where $\gamma$ is the Euler-Mascheroni constant. Combining~\cref{eq:aar_Obj_max_relation} and~\cref{eq:aar_max_expectation},
\begin{equation}\label{eq:aar_obj_tau_form}
\mathbb{E}\big[G\cdot q^\tau(G)\big]
= \gamma + \log\Big(\sum_{v} e^{\tau_v}\Big) - \tau\cdot r.
\end{equation}
Yet, we have 
\begin{equation} \label{eq:aar_tau_dot_r}
    \tau\cdot r
= \sum_{u} r_u \log r_u + \log\Big(\sum_{v\in\Sigma} e^{\tau_v}\Big).
\end{equation}
Substituting~\cref{eq:aar_tau_dot_r} into~\cref{eq:aar_obj_tau_form} yields
\begin{equation*}
\mathbb{E}\big[G\cdot q^\tau(G)\big]
= \gamma - \sum_{u\in\Sigma} r_u \log r_u
= \gamma + H(r),
\end{equation*}
where $H(r):=-\sum_{u}r_u\log r_u$ is the Shannon entropy. Thus, in the Gumbel case, optimizing~\cref{eq:aar_optimization_problem} over $q$ is equivalent to optimizing over $r=\mathbb{E}[q(G)]$:
\begin{align}
\max_{r\in\Delta(\Sigma)}\quad & H(r)\label{eq:aar_entropy_problem_obj}\\
\text{s.t.}\quad & \mathrm{KL}(r\|p)\le\varepsilon.\label{eq:aar_entropy_problem_constr}
\end{align}

Introduce multipliers $\lambda\ge 0$ for~\cref{eq:aar_entropy_problem_constr} and $\beta\in\mathbb{R}$ for $\sum_{u}r_u=1$. The Lagrangian is
\begin{equation}\label{eq:aar_entropy_L}
\mathcal{J}(r,\lambda,\beta)
= -\sum_{u} r_u\log r_u
- \lambda\left(\sum_{u}r_u\log\frac{r_u}{p_u}-\varepsilon\right)
+ \beta\left(\sum_{u}r_u-1\right).
\end{equation}
Assuming $r_u>0$, stationarity in $r_u$ gives
\begin{equation*}
-(\log r_u+1) - \lambda\big(\log(r_u/p_u)+1\big) + \beta = 0,\qquad u\in\Sigma.
\end{equation*}
Thus
\begin{equation*}
r_u = e^{C} p_u^{\alpha},
\end{equation*}
where
\begin{equation}\label{eq:aar_alpha_C_def}
\alpha := \frac{\lambda}{1+\lambda}\in[0,1),\qquad
C := \frac{\beta-(1+\lambda)}{1+\lambda}.
\end{equation}
Normalization $\sum_{u}r_u=1$ fixes $e^{C}$, and we obtain the family
\begin{equation}\label{eq:aar_r_alpha}
r_u^*(\alpha)
= \frac{p_u^\alpha}{\sum_{v\in\Sigma}p_v^\alpha},\qquad \alpha\in[0,1).
\end{equation}

Note that $r^*(0)$ is the uniform distribution on $\Sigma$, and $r^*(1)=p$. The map $\alpha\mapsto \mathrm{KL}(r^*(\alpha)\|p)$ is continuous, with
\begin{equation*}
\mathrm{KL}\big(r^*(0)\|p\big)=\mathrm{KL}(\mathrm{Unif}\|p),\qquad
\mathrm{KL}\big(r^*(1)\|p\big)=0.
\end{equation*}
Hence:
\begin{itemize}
\item If $\varepsilon\ge \mathrm{KL}(\mathrm{Unif}\|p)$, the KL constraint is slack and the unconstrained maximizer of $H(r)$ is the uniform distribution $r^*(0)$, corresponding to $\alpha=0$.
\item If $0\le\varepsilon\le\mathrm{KL}(\mathrm{Unif}\|p)$, there exists (by continuity) some $\alpha\in[0,1]$ such that
\begin{equation*}
\mathrm{KL}\big(r^*(\alpha)\|p\big)=\varepsilon,
\end{equation*}
and the constraint is active at the optimum.
\end{itemize}

For any optimal solution $r^*$ of~\cref{eq:aar_entropy_problem_obj}--\cref{eq:aar_entropy_problem_constr}, the softmax bijection~\cref{eq:aar_softmax} allows us to realize $r^*$ as $r(\tau)$ for some $\tau$. From~\cref{eq:aar_r_alpha}, we have
\begin{equation*}
\log r_u^*(\alpha) = \alpha\log p_u + \text{const},\qquad u\in\Sigma.
\end{equation*}
Hence, we see that $\tau_u$ can be chosen of the form
\begin{equation*}
\tau_u = \alpha\log p_u + c,\qquad u\in\Sigma,
\end{equation*}
for some constant $c\in\mathbb{R}$. Adding a constant to all coordinates of $\tau$ does not change the argmax in~\cref{eq:aar_q_tau}, so without loss of generality we may take $c=0$, that is,
\begin{equation*}
\tau_u = \alpha\log p_u,\qquad u\in\Sigma.
\end{equation*}
Combining this with~\cref{eq:aar_q_tau} yields
\begin{equation*}
q^*(G)_u
= \mathbbm{1}\Big\{u=\arg\max_{v\in\Sigma}\big(G_v+\alpha\log p_v\big)\Big\},
\end{equation*}
for some $\alpha\in[0,1]$ chosen as above.
This completes the proof.
\end{proof}

\begin{theorem}\label{theorem:aar_watermarks_penalized}
Let $p\in\Delta(\Sigma)$ have full support, let $G$ be a random vector with support $\mathbb{G}\subseteq\mathbb{R}^{|\Sigma|}$ and i.i.d.\ coordinates from a Gumbel distribution, and let $\delta > 0$. Consider
\begin{align*}
q^* \in \argmax_{q:\mathbb{G}\to\Delta(\Sigma)} \quad & \mathbb{E}\big[G\cdot q(G)\big] - \frac{1}{\delta}\mathrm{KL}\big(\mathbb{E}[q(G)]\|p\big),
\end{align*}
where $\mathrm{KL}(r\|p)=\sum_{u\in\Sigma} r_u\log\frac{r_u}{p_u}$.
Then, we have
\begin{equation*}
\forall u\in\Sigma,\qquad q^*(G)_u \;=\; \mathbbm{1}\Big\{u\in\arg\max_{v\in\Sigma}\big(G_v+\frac{1}{1+\delta} \log p_v\big)\Big\}.
\end{equation*}
\end{theorem}

\begin{proof}
    The proof is similar to that of \cref{theorem:aar_watermarks}.
    In \cref{eq:aar_entropy_L}, instead of being a Lagrangian variable, $\lambda$ is explicitly given as $\lambda = \frac{1}{\delta}$.
    This then gives, in \cref{eq:aar_alpha_C_def}, $\alpha = \frac{1}{1 + \delta}$.
    Thus, we have that,
    \begin{equation*}
        \forall u\in\Sigma,\qquad q^*(G)_u \;=\; \mathbbm{1}\Big\{u\in\arg\max_{v\in\Sigma}\big(G_v+\frac{1}{1+\delta} \log p_v\big)\Big\},
    \end{equation*}
    which completes the proof.
\end{proof}

\begin{theorem}\label{theorem:distortion_free_argmax}
Let $p\in\Delta(\Sigma)$, let $G$ be a random vector with support $\mathbb{G}\subseteq\mathbb{R}^{|\Sigma|}$ and i.i.d.\ coordinates from a continuous law $\mathcal{L}$ with finite first-order moment, and let $\varepsilon\ge0$. Consider
\begin{align*}
q^* \in \argmax_{q:\mathbb{G}\to\Delta(\Sigma)} \quad & \mathbb{E}\big[G\cdot q(G)\big] \\
\text{subject to}\quad & D\big(\mathbb{E}[q(G)], p \big)\le \varepsilon,
\end{align*}
where $D: \Delta(\Sigma) \times \Delta(\Sigma) \rightarrow \mathbb{R}$.
Then, if the problem is feasible, there exists an optimal solution with support of size $1$.
\end{theorem}

\begin{proof}
    Let $q$ be any optimal solution. 
    Let $X\in\Sigma$ be a random variable corresponding to the sampled token, and define
    \begin{equation*}
        \forall g \in \mathbb{G},\ \forall u \in \Sigma,\quad \mathbb{P}[X = u \mid G = g] = q_u(g).
    \end{equation*}
    Then the objective rewrites as
    \begin{equation*}
        \mathbb{E}[G \cdot q(G)] = \mathbb{E}[G_X],
    \end{equation*}
    and the average output distribution is the marginal of $X$,
    \begin{equation*}
        r(q)_u \coloneqq \mathbb{E}[q_u(G)] = \mathbb{P}[X = u],\quad u\in\Sigma.
    \end{equation*}
    The law of $G$ is atomless since its coordinates are i.i.d.\ from a continuous law. By the Dvoretsky-Wald-Wolfowitz purification theorem~\citep{DWW_purification_theorem} applied to $g\mapsto q(g)$ and the finite family of integrable functions
    \begin{equation*}
        (g,u)\mapsto g_u \quad\text{and}\quad (g,u)\mapsto \mathbf{1}\{u = v\},\ v\in\Sigma,
    \end{equation*}
    there exists a measurable function $f:\mathbb{G}\to\Sigma$ such that, with $X' = f(G)$,
    \begin{equation*}
        \mathbb{E}[G_{X'}] = \mathbb{E}[G_X]
        \quad\text{and}\quad
        \mathbb{P}[X' = u] = \mathbb{P}[X = u],\ \forall u\in\Sigma.
    \end{equation*}
    Define $q':\mathbb{G}\to\Delta(\Sigma)$ by $q'(g) = e_{f(g)}$, so that $q'(G)$ has support of size $1$ almost surely. Then
    \begin{equation*}
        \mathbb{E}[G\cdot q'(G)] = \mathbb{E}[G_{X'}] = \mathbb{E}[G_X] = \mathbb{E}[G\cdot q(G)],
    \end{equation*}
    and
    \begin{equation*}
        \mathbb{E}[q'(G)]_u = \mathbb{P}[X'=u] = \mathbb{P}[X=u] = \mathbb{E}[q(G)]_u,\quad u\in\Sigma.
    \end{equation*}
    Hence $D(\mathbb{E}[q'(G)],p) = D(\mathbb{E}[q(G)],p)\le \varepsilon$, so $q'$ is feasible and optimal. This proves the existence of an optimal solution with support of size $1$.
\end{proof}

\subsection{Chi2 Watermarks}

Here we formalize and then prove the various results established for the one-layer SynthID and $\chi^2$ watermark.  
In particular, we first establish the general watermark sampling mechanism for the $\chi^2$ watermark in \cref{theorem:chi2_watermarks}, and show in \cref{theorem:chi2_watermarks_penalized} how it reduces to the 1-layer SynthID tournament in some specific cases.

\begin{theorem}\label{theorem:chi2_watermarks}
Let $p\in\Delta(\Sigma)$ have full support, let $g\in\mathbb{G}$, and let $\varepsilon>0$. Consider
\begin{align*}
q^*(g) \in \argmax_{q(g)\in\Delta(\Sigma)} \quad & g\cdot q(g) \\
\text{subject to}\quad & \sum_{u \in \Sigma} \frac{(q(g)_u - p_u)^2}{p_u} \le \varepsilon.
\end{align*}

Then there exists $\delta, \mu \in \mathbb{R}$ such that,
\begin{equation*}
    \forall u \in \Sigma, q^*(g)_u = p_u [1 + \delta(g_u + \mu)]_+
\end{equation*}
where $[x]_+ \coloneq \max(0,x)$, $\mu$ is chosen such that $\sum_{u \in \Sigma} q^*(g)_u = 1$, and $\delta$ is chosen such that $\sum_{u \in \Sigma} \frac{(q^*(g)_u - p_u)^2}{p_u} = \varepsilon$ if feasible, else the greedy solution is feasible.
\end{theorem}

\begin{proof}
    
    We introduce a Lagrange multiplier $\lambda \ge 0$ for the inequality constraint $\sum_{u \in \Sigma} \frac{(q_u - p_u)^2}{p_u} = \varepsilon$, a multiplier $\mu \in \mathbb{R}$ for the simplex constraint $\sum_u q_u = 1$ and $\nu \in \mathbb{R}^{|\Sigma|}_+$ for the positivity constraints. 
    The Lagrangian is:
    \begin{equation*}
        \mathcal{L}(q, \lambda, \mu) = -\sum_{u\in\Sigma} q_u g_u + \lambda \left( \sum_{u\in\Sigma} \frac{(q_u - p_u)^2}{p_u} - \varepsilon \right) - \mu \left( \sum_{u\in\Sigma} q_u - 1 \right) - \sum_{u \in \Sigma} \nu_u q_u.
    \end{equation*}
    Note that we minimize $-\mathcal{L}$ to maximize the objective. Taking the derivative with respect to $q_u$ and setting it to zero:
    \begin{equation*}
        \frac{\partial \mathcal{L}}{\partial q_u} = -g_u + 2\lambda  \frac{(q_u - p_u)}{p_u} - \mu - \nu_u = 0.
    \end{equation*}
    Rearranging for $q_u$:
    \begin{equation*}
        q_u = p_u \left(1 + \frac{1}{2\lambda}(g_u + \mu + \nu_u) \right).
    \end{equation*}
    In particular, we have $\forall u \in \Sigma, - q_u \nu_u = 0$ with $\nu_u \ge 0$, which means that,
    \begin{equation*}
        q_u = p_u \left[1 + \frac{1}{2\lambda}(g_u + \mu) \right]_+.
    \end{equation*}
    By letting $\delta = \frac{1}{2\lambda}$ we find the shape from \cref{theorem:chi2_watermarks}.
\end{proof}

\begin{theorem}\label{theorem:chi2_watermarks_penalized}
Let $p\in\Delta(\Sigma)$ have full support, let $g\in\mathbb{G}$, and let $\delta>0$. Consider
\begin{align*}
q^*(g) \in \argmax_{q(g)\in\Delta(\Sigma)} \quad & g\cdot q(g) - \frac{1}{2\delta}\sum_{u \in \Sigma} \frac{(q(g)_u - p_u)^2}{p_u}.
\end{align*}

Then there exists $\mu \in \mathbb{R}$ such that,
\begin{equation*}
    \forall u \in \Sigma, q^*(g)_u = p_u [1 + \delta(g_u + \mu)]_+
\end{equation*}
where $\mu$ is chosen such that $\sum_{u \in \Sigma} q^*(g)_u = 1$. 

In particular, if $\delta \le \frac{1}{\max_{u\in \Sigma} g_u - \min_{u\in \Sigma} g_u}$, then we have
\begin{equation*}
    \forall u \in \Sigma, q^*(g)_u = p_u (1 + \delta(g_u - p \cdot g))
\end{equation*}
\end{theorem}

\begin{proof}
    The derivation of the general form is identical to the first part of \cref{theorem:chi2_watermarks}, where $\delta$ is treated as a fixed parameter rather than a Lagrange multiplier.

    To prove the second part, assume that the non-negativity constraints are inactive, i.e., for all $u \in \Sigma$, $1 + \delta(g_u + \mu) \ge 0$. 
    Summing the probability mass constraint $\sum q_u = 1$ implies:
    \begin{equation*}
        \sum_{u \in \Sigma} p_u [1 + \delta(g_u + \mu)] = 1 \implies 1 + \delta(p \cdot g + \mu) = 1.
    \end{equation*}
    Since $\delta > 0$, this yields $\mu = - p \cdot g$. Substituting this back into the solution gives the candidate form $q^*(g)_u = p_u (1 + \delta(g_u - p \cdot g))$.

    We must strictly verify that this candidate satisfies the non-negativity assumption, i.e.,
    \begin{equation*}
        \forall u \in \Sigma, \quad 1 + \delta(g_u - p \cdot g) \ge 0 \iff \delta(p \cdot g - g_u) \le 1.
    \end{equation*}
    The tightest constraint occurs when $p \cdot g - g_u$ is maximized, which corresponds to $g_u = \min_{v} g_v$. Thus, we require $\delta \le \frac{1}{p \cdot g - \min_{v} g_v}$.
    
    Observing that $p \cdot g \le \max_{v} g_v$, we have:
    \begin{equation*}
         p \cdot g - \min_{v} g_v \le \max_{v} g_v - \min_{v} g_v.
    \end{equation*}
    Therefore, the condition $\delta \le \frac{1}{\max_{u} g_u - \min_{u} g_u}$ is sufficient to guarantee feasibility.
\end{proof}

\subsection{SynthID Watermark}

In this part, we extend our framework to capture the SynthID watermark.
In particular, we show that SynthID with $m-layer$ tournament corresponds to chaining $m$ optimization problems with $\chi^2$ constraints.

\begin{theorem}\label{theorem:synthid_text}
Let $p\in\Delta(\Sigma)$ have full support, let $m \in \mathbb{N}$, let $g^{(1)},\dots,g^{(m)}\in \{0,1\}^{|\Sigma|}$. 
Consider the chained optimization problems over $q^{(1)},\dots,q^{(m)} \in \Delta(\Sigma)$,
\begin{equation*} q^{(i)} \in \argmax_{q \in \Delta(\Sigma)} \quad g^{(i)}\cdot q - \frac{1}{2}\sum_{u \in \Sigma} \frac{(q_u - q^{(i-1)}_u)^2}{q^{(i-1)}_u}. 
\end{equation*}
with $q^{(0)} \coloneq p$. Then, we have that the family defined by,
\begin{equation*} 
    \forall i \in \{1,\dots,m\}, \forall u \in \Sigma, q^{(i)}_u = q^{(i-1)}_u (1 + (g^{(i)}_u - q^{(i-1)}_u \cdot g^{(i)}))
\end{equation*}
is an optimal solution.
\end{theorem}

\begin{proof}
    Simply proceed by induction. For each step $i$, \cref{theorem:chi2_watermarks_penalized} holds, since we have $\delta = 1 \le \frac{1}{1-0}$.
\end{proof}

\subsection{PPL Watermarks}
\label{app:proofs:ppl_watermarks}

Here, in \cref{theorem:soft_ppl_watermark}, we formalize and then prove the watermark sampling mechanism established for the soft PPL watermark.  
For the hard PPL watermark, the resulting constrained optimization problem is a linear programming problem and thus all results stated in \cref{sec:method_practice:new} are standard results from the linear programming literature.

\begin{theorem}\label{theorem:soft_ppl_watermark}
    Let $p\in\Delta(\Sigma)$ have full support, let $G$ be a random vector with support $\mathbb{G}\subseteq\mathbb{R}^{|\Sigma|}$ and i.i.d.\ coordinates from a continuous law $\mathcal{L}$, and let $\varepsilon\ge0$. Consider the optimization problem
    \begin{align*}
        q^* \in \argmax_{q:\mathbb{G}\to\Delta(\Sigma)} \quad & \mathbb{E}\big[G\cdot q(G)\big] \\
        \text{subject to}\quad & (p - \mathbb{E}_G[q(G)]) \cdot \log p  \le \varepsilon.
    \end{align*}
    Then, an optimal solution is given by
    \begin{equation*}
        \forall g \in \mathbb{G}, q^*(g) = \arg\max(g_{i}+\lambda \log p_{i}),
    \end{equation*}
    where $\lambda$ is chosen such that $(p - \mathbb{E}_G[q(G)]) \cdot \log p = \varepsilon$ if feasible, else $\lambda =0$.
\end{theorem}

\begin{proof}
By \cref{theorem:distortion_free_argmax}, there exists an optimal solution $q$ with support of size $1$ a.s., i.e.\ there is a measurable $f:\mathbb{G}\to\Sigma$ such that
\[
q(g) = e_{f(g)},\qquad \forall g\in\mathbb{G}.
\]

We now solve the problem over such deterministic $q$. 
We introduce a Lagrange multiplier $\lambda\ge 0$ for the constraint. 
The Lagrangian is
\begin{align}
    \mathcal{L}(q,\lambda)
    &= \mathbb{E}[G\cdot q(G)]
      - \lambda\big((p-\mathbb{E}[q(G)])\cdot\log p - \varepsilon\big)\\
    &= \mathbb{E}\big[(G+\lambda \log p)\cdot q(G)\big]
      - \lambda\big(p\cdot\log p - \varepsilon\big). \label{eq:soft_ppl_Lagrangian}
\end{align}
For fixed $\lambda$, the term in \cref{eq:soft_ppl_Lagrangian} is maximized by
\begin{equation*}
    f_\lambda(g) \in \arg\max_{i\in\Sigma}\big(g_i + \lambda\log p_i\big).
\end{equation*}

If the unconstrained maximizer (corresponding to $\lambda=0$) already satisfies $(p-\mathbb{E}[q^0(G)])\cdot\log p\le\varepsilon$, then $\lambda=0$.  
Otherwise the constraint is active at the optimum, so $\lambda>0$ and 
\begin{equation*}
    (p-\mathbb{E}[e_{f_\lambda(G)}])\cdot\log p = \varepsilon.
\end{equation*}
This concludes the proof.
\end{proof}

%% file: src/appendix/examples.tex
\section{Watermarked Text Examples}
\label{app:text_example}

This section presents a sample user prompt together with the corresponding watermarked responses. 
We also report the associated p-values and perplexity scores, where perplexity is computed using \qwen{}.
The formatting of the LLM response has been lightly edited for readability.
Each response is 200 tokens long.

\input{prompts/example1.tex}

%% file: prompts/example1.tex
\begin{promptbox}
    Why can't we just taste candy or Sweets and then spit it out to avoid its unhealthy attributes? What makes us swallow it to get satisfaction?
\end{promptbox}
\begingroup\raggedleft
\begin{modelbox}[AAR, $\, p=7.4e-5$, $\, PPL=4.8$]{outerframe}{outerbg}
    You're right that our brains are wired to enjoy the taste of sweets, but there are several reasons why we tend to swallow them, even though we know they're unhealthy. Here are a few possible explanations:
    \\\\
1. \textbf{Evolutionary adaptations}: For millions of years, humans relied on sugar-rich foods like fruits and berries for energy and survival. As a result, our brains developed a preference for sweet tastes, which signaled the presence of energy-rich compounds. In many cases, it was beneficial to consume these foods and store the energy for later use. Our brains haven't fully adjusted to the modern environment, where sugary foods are abundant but not necessarily necessary for survival.
\\
2. \textbf{Reward processing}: Sugar activates the brain's reward centers, releasing feel-good chemicals like dopamine, endorphins, and opiids. These sensations create a strong psychological association between sweetness and pleasure, reinforcing the behavior of consuming sugary foods. Our brains are wired to seek out experiences that stimulate these reward pathways, even if
\end{modelbox}

\begin{modelbox}[Chi2, $\, p=0.0$, $\, PPL=5.7$]{outerframe}{outerbg}
    You're right that it would be simpler if our body automatically spat out foods that aren't beneficial to us, but that's not how our digestive system is designed. Here are a few reasons why:
\\\\
1. \textbf{Evolutionary design:} The human body has an amazing ability to adapt and learn, but it's based on the principle of survival and reproduction. Our digestive system has evolved over millions of years to process a wide range of foods, including those that are high in sugar and other energy-rich nutrients. These foods were once scarce, so our bodies were designed to rapidly absorb their nutrients for energy. This allowed our ancestors to survive and thrive in times of scarcity.
\\
2. \textbf{The role of the brain:} Our brain plays a significant role in controlling what we eat and how our body processes it. The pleasure and satisfaction we experience when eating something sweet, like candy, is mediated by a system
\end{modelbox}

\begin{modelbox}[KGW, $\, p=0.0$, $\, PPL=6.4$]{outerframe}{outerbg}
    You're right that we can choose to spit out candy or sweets if we know they're unhealthy. However, the reason we often can't help but swallow and enjoy these treats is deeply rooted in human psychology, neuroscience, and biology.

Here are some possible explanations:
\\\\
1. \textbf{Evolutionary history}: Throughout human evolution, sweet tastes were a signal of nutrient-dense foods that were scarce and valuable. Our brains evolved to associate sweet tastes with pleasure and reward, encouraging us to consume them and gain energy. In modern times, we still respond to sweet tastes as a cue for pleasure, but our diets no longer match the evolutionary landscape.\\
2. \textbf{Dopamine release}: When we eat something sweet, our brains release dopamine, a neurotransmitter associated with pleasure, reward, and motivation. This release can be intense and immediate, making it difficult to resist the desire to consume more.\\
3. \textbf{Conditioned response}: We often learn to associate certain foods, including sweets, with pleasure and satisfaction
\end{modelbox}

\begin{modelbox}[PPL (Hard), $\, p=1.1e-9$, $\, PPL=7.4$]{outerframe}{outerbg}
    It's indeed intriguing as to why our brains would lead us to consume sweets, knowing their potential harm, and yet, we can't help but crave them. There are several biological, psychological, and historical reasons behind this phenomenon.
\\\\
\textbf{Biological reasons:}
\\\\
1. \textbf{Evolutionary legacy:} In the past, sweets were rare and often associated with high-calorie energy sources, making them valuable for survival. Our brains were adapted to consume these energy-dense foods to store fat for potential famines. This instinctual response has been retained, although no longer necessary in modern times with an abundance of food available.\\
2. \textbf{Dopamine release:} Consuming sweets and other pleasurable foods triggers the release of dopamine, a neurotransmitter associated with pleasure signaling. The brain's reward system is activated due to the combination of sugar, fat, and other nutrients releasing dopamine, leading a craving for more.\\
3. \textbf{Lingering taste after spittle:} Some studies indicate even a small amount
\end{modelbox}

\begin{modelbox}[PPL (Soft), $\, p=1.4e-3$, $\, PPL=5.4$]{outerframe}{outerbg}
    You're right that it would be ideal if we could just taste sweets and spit them out, avoiding the negative effects of high sugar and calorie intake. However, the process of eating and swallowing is complex, and it's influenced by both physiological and psychological factors. Here are some reasons why we tend to swallow sweets:\\\\

1. \textbf{Evolutionary origins:} In the past, eating sweet-tasting foods indicated the presence of energy-rich carbohydrates, which were essential for survival. Our ancestors would consume these foods to sustain themselves, especially during times of scarcity. This evolutionary adaptation has been passed down through generations, making us more likely to associate sweet tastes with energy and satisfaction.\\
2. \textbf{Brain chemistry:} When we eat sweets, they activate the brain's reward centers, releasing feelgood hormones like dopamine. This triggers a sense of pleasure and satisfaction, which can lead to overconsumption. The brain's reward system is designed to reinforce behaviors that provide energy and survival benefits, making us more likely to repeat them
\end{modelbox}

\begin{modelbox}[SynthID, $\, p=0.0$, $\, PPL=13.9$]{outerframe}{outerbg}
    You're right that it'd seem logical to just taste and spit out candy or sweets to avoid their potential negative effects on our health. However, our bodies have complex neurological and physiological mechanisms that play a significant role in our eating behavior.\\\\

Here are some reasons why we tend to swallow the candy and experience the satisfaction we're talking about:\\\\

1. \textbf{Evolutionary mechanisms}: Throughout history, our brains were designed to respond quickly to available energy sources, like sugary or fatty foods, as these were essential for survival. In the case of sweets, our brain is programmed to respond strongly to glucose (sugar), releasing the neurotransmitter dopamine in the process. This associative reward network drives us to repeat behaviors associated with pleasure, making us reach out for those sweet treats again.

2. \textbf{Familiar texture and taste}: When you put something pleasant in your mouth, the sensory experience is unique in several ways: It satisfies your sweet tooth by releasing dopamine in your brain. Our brains are set up to
\end{modelbox}
\endgroup

%% file: main.bbl
\begin{thebibliography}{30}
\providecommand{\natexlab}[1]{#1}
\providecommand{\url}[1]{\texttt{#1}}
\expandafter\ifx\csname urlstyle\endcsname\relax
  \providecommand{\doi}[1]{doi: #1}\else
  \providecommand{\doi}{doi: \begingroup \urlstyle{rm}\Url}\fi

\bibitem[Aaronson(2023)]{aar}
Aaronson, S.
\newblock Watermarking of large language models.
\newblock In \emph{Workshop on Large Language Models and Transformers, Simons Institute, UC Berkeley}, 2023.

\bibitem[Alihosseini et~al.(2019)Alihosseini, Montahaei, and Baghshah]{self_bleu}
Alihosseini, D., Montahaei, E., and Baghshah, M.~S.
\newblock Jointly measuring diversity and quality in text generation models.
\newblock In \emph{Proceedings of the Workshop on Methods for Optimizing and Evaluating Neural Language Generation}, pp.\  90--98, 2019.

\bibitem[Austin et~al.(2021)Austin, Odena, Nye, Bosma, Michalewski, Dohan, Jiang, Cai, Terry, Le, et~al.]{mbpp}
Austin, J., Odena, A., Nye, M., Bosma, M., Michalewski, H., Dohan, D., Jiang, E., Cai, C., Terry, M., Le, Q., et~al.
\newblock Program synthesis with large language models.
\newblock \emph{arXiv preprint arXiv:2108.07732}, 2021.

\bibitem[Chen et~al.(2021)Chen, Tworek, Jun, Yuan, de~Oliveira~Pinto, Kaplan, Edwards, Burda, Joseph, Brockman, Ray, Puri, Krueger, Petrov, Khlaaf, Sastry, Mishkin, Chan, Gray, Ryder, Pavlov, Power, Kaiser, Bavarian, Winter, Tillet, Such, Cummings, Plappert, Chantzis, Barnes, Herbert-Voss, Guss, Nichol, Paino, Tezak, Tang, Babuschkin, Balaji, Jain, Saunders, Hesse, Carr, Leike, Achiam, Misra, Morikawa, Radford, Knight, Brundage, Murati, Mayer, Welinder, McGrew, Amodei, McCandlish, Sutskever, and Zaremba]{humaneval}
Chen, M., Tworek, J., Jun, H., Yuan, Q., de~Oliveira~Pinto, H.~P., Kaplan, J., Edwards, H., Burda, Y., Joseph, N., Brockman, G., Ray, A., Puri, R., Krueger, G., Petrov, M., Khlaaf, H., Sastry, G., Mishkin, P., Chan, B., Gray, S., Ryder, N., Pavlov, M., Power, A., Kaiser, L., Bavarian, M., Winter, C., Tillet, P., Such, F.~P., Cummings, D., Plappert, M., Chantzis, F., Barnes, E., Herbert-Voss, A., Guss, W.~H., Nichol, A., Paino, A., Tezak, N., Tang, J., Babuschkin, I., Balaji, S., Jain, S., Saunders, W., Hesse, C., Carr, A.~N., Leike, J., Achiam, J., Misra, V., Morikawa, E., Radford, A., Knight, M., Brundage, M., Murati, M., Mayer, K., Welinder, P., McGrew, B., Amodei, D., McCandlish, S., Sutskever, I., and Zaremba, W.
\newblock Evaluating large language models trained on code.
\newblock 2021.

\bibitem[Chen et~al.(2025)Chen, Wu, Guo, and Huang]{mcmark}
Chen, R., Wu, Y., Guo, J., and Huang, H.
\newblock Improved unbiased watermark for large language models.
\newblock \emph{CoRR}, 2025.

\bibitem[Christ et~al.(2024)Christ, Gunn, and Zamir]{undetectable}
Christ, M., Gunn, S., and Zamir, O.
\newblock Undetectable watermarks for language models.
\newblock In \emph{The Thirty Seventh Annual Conference on Learning Theory}, pp.\  1125--1139. PMLR, 2024.

\bibitem[Cobbe et~al.(2021)Cobbe, Kosaraju, Bavarian, Chen, Jun, Kaiser, Plappert, Tworek, Hilton, Nakano, Hesse, and Schulman]{gsm8k}
Cobbe, K., Kosaraju, V., Bavarian, M., Chen, M., Jun, H., Kaiser, L., Plappert, M., Tworek, J., Hilton, J., Nakano, R., Hesse, C., and Schulman, J.
\newblock Training verifiers to solve math word problems.
\newblock \emph{arXiv preprint arXiv:2110.14168}, 2021.

\bibitem[Dathathri et~al.(2024)Dathathri, See, Ghaisas, Huang, McAdam, Welbl, Bachani, Kaskasoli, Stanforth, Matejovicova, et~al.]{synthid}
Dathathri, S., See, A., Ghaisas, S., Huang, P.-S., McAdam, R., Welbl, J., Bachani, V., Kaskasoli, A., Stanforth, R., Matejovicova, T., et~al.
\newblock Scalable watermarking for identifying large language model outputs.
\newblock \emph{Nature}, 634\penalty0 (8035):\penalty0 818--823, 2024.

\bibitem[Dvoretzky et~al.(1951)Dvoretzky, Wald, and Wolfowitz]{DWW_purification_theorem}
Dvoretzky, A., Wald, A., and Wolfowitz, J.
\newblock {Elimination of Randomization in Certain Statistical Decision Procedures and Zero-Sum Two-Person Games}.
\newblock \emph{The Annals of Mathematical Statistics}, 22\penalty0 (1):\penalty0 1 -- 21, 1951.
\newblock \doi{10.1214/aoms/1177729689}.
\newblock URL \url{https://doi.org/10.1214/aoms/1177729689}.

\bibitem[Fan et~al.(2019)Fan, Jernite, Perez, Grangier, Weston, and Auli]{eli5}
Fan, A., Jernite, Y., Perez, E., Grangier, D., Weston, J., and Auli, M.
\newblock Eli5: Long form question answering, 2019.
\newblock URL \url{https://arxiv.org/abs/1907.09190}.

\bibitem[Fernandez et~al.(2023)Fernandez, Chaffin, Tit, Chappelier, and Furon]{three_bricks}
Fernandez, P., Chaffin, A., Tit, K., Chappelier, V., and Furon, T.
\newblock Three bricks to consolidate watermarks for large language models.
\newblock \emph{2023 IEEE International Workshop on Information Forensics and Security (WIFS)}, 2023.

\bibitem[Gao et~al.(2024)Gao, Tow, Abbasi, Biderman, Black, DiPofi, Foster, Golding, Hsu, Le~Noac'h, Li, McDonell, Muennighoff, Ociepa, Phang, Reynolds, Schoelkopf, Skowron, Sutawika, Tang, Thite, Wang, Wang, and Zou]{eval-harness}
Gao, L., Tow, J., Abbasi, B., Biderman, S., Black, S., DiPofi, A., Foster, C., Golding, L., Hsu, J., Le~Noac'h, A., Li, H., McDonell, K., Muennighoff, N., Ociepa, C., Phang, J., Reynolds, L., Schoelkopf, H., Skowron, A., Sutawika, L., Tang, E., Thite, A., Wang, B., Wang, K., and Zou, A.
\newblock The language model evaluation harness, 07 2024.
\newblock URL \url{https://zenodo.org/records/12608602}.

\bibitem[Gloaguen et~al.(2025{\natexlab{a}})Gloaguen, Jovanovi{\'c}, Staab, and Vechev]{gloaguen2025blackbox}
Gloaguen, T., Jovanovi{\'c}, N., Staab, R., and Vechev, M.
\newblock Black-box detection of language model watermarks.
\newblock In \emph{The Thirteenth International Conference on Learning Representations}, 2025{\natexlab{a}}.
\newblock URL \url{https://openreview.net/forum?id=E4LAVLXAHW}.

\bibitem[Gloaguen et~al.(2025{\natexlab{b}})Gloaguen, Staab, Jovanović, and Vechev]{gloaguen2025watermarkingdiffusionlanguagemodels}
Gloaguen, T., Staab, R., Jovanović, N., and Vechev, M.
\newblock Watermarking diffusion language models, 2025{\natexlab{b}}.
\newblock URL \url{https://arxiv.org/abs/2509.24368}.

\bibitem[He et~al.(2024)He, Liu, Wang, Mao, and Bu]{dawa}
He, H., Liu, Y., Wang, Z., Mao, Y., and Bu, Y.
\newblock Theoretically grounded framework for llm watermarking: A distribution-adaptive approach.
\newblock \emph{arXiv preprint arXiv:2410.02890}, 2024.

\bibitem[Hou et~al.(2024)Hou, Zhang, He, Wang, Chuang, Wang, Shen, Van~Durme, Khashabi, and Tsvetkov]{semstamp}
Hou, A., Zhang, J., He, T., Wang, Y., Chuang, Y.-S., Wang, H., Shen, L., Van~Durme, B., Khashabi, D., and Tsvetkov, Y.
\newblock Semstamp: A semantic watermark with paraphrastic robustness for text generation.
\newblock In \emph{Proceedings of the 2024 Conference of the North American Chapter of the Association for Computational Linguistics: Human Language Technologies (Volume 1: Long Papers)}, pp.\  4067--4082, 2024.

\bibitem[Huang et~al.(2023)Huang, Zhu, Zhu, Lee, Jiao, and Jordan]{optimal_statistical_wm}
Huang, B., Zhu, B., Zhu, H., Lee, J., Jiao, J., and Jordan, M.
\newblock Towards optimal statistical watermarking.
\newblock In \emph{Socially Responsible Language Modelling Research}, 2023.
\newblock URL \url{https://openreview.net/forum?id=Fc2FaS9mYJ}.

\bibitem[Kirchenbauer et~al.(2023)Kirchenbauer, Geiping, Wen, Katz, Miers, and Goldstein]{kgw1}
Kirchenbauer, J., Geiping, J., Wen, Y., Katz, J., Miers, I., and Goldstein, T.
\newblock A watermark for large language models.
\newblock In \emph{International Conference on Machine Learning}, pp.\  17061--17084. PMLR, 2023.

\bibitem[Kodama et~al.(2025)Kodama, Diwan, Rosenblatt, Witter, and Cohen]{simkey}
Kodama, S., Diwan, H., Rosenblatt, L., Witter, R.~T., and Cohen, N.
\newblock Simkey: A semantically aware key module for watermarking language models, 2025.
\newblock URL \url{https://arxiv.org/abs/2510.12828}.

\bibitem[Kuditipudi et~al.(2024)Kuditipudi, Thickstun, Hashimoto, and Liang]{kth}
Kuditipudi, R., Thickstun, J., Hashimoto, T., and Liang, P.
\newblock Robust distortion-free watermarks for language models.
\newblock \emph{TMLR}, 2024.

\bibitem[Li et~al.(2025{\natexlab{a}})Li, Li, and Zhang]{likelihood_wm_detection}
Li, X., Li, G., and Zhang, X.
\newblock A likelihood based approach for watermark detection.
\newblock In \emph{The 28th International Conference on Artificial Intelligence and Statistics}, 2025{\natexlab{a}}.

\bibitem[Li et~al.(2025{\natexlab{b}})Li, Ruan, Wang, Long, and Su]{statistical_wm}
Li, X., Ruan, F., Wang, H., Long, Q., and Su, W.~J.
\newblock A statistical framework of watermarks for large language models: Pivot, detection efficiency and optimal rules.
\newblock \emph{The Annals of Statistics}, 53\penalty0 (1):\penalty0 322--351, 2025{\natexlab{b}}.

\bibitem[Ren et~al.(2024)Ren, Xu, Liu, Cui, Wang, Yin, and Tang]{semamark}
Ren, J., Xu, H., Liu, Y., Cui, Y., Wang, S., Yin, D., and Tang, J.
\newblock A robust semantics-based watermark for large language model against paraphrasing.
\newblock In \emph{Findings of the Association for Computational Linguistics: NAACL 2024}, pp.\  613--625, 2024.

\bibitem[Shen et~al.()Shen, Huang, and Wan]{scrubbing_def1}
Shen, H., Huang, B., and Wan, X.
\newblock Enhancing llm watermark resilience against both scrubbing and spoofing attacks.
\newblock In \emph{The Thirty-ninth Annual Conference on Neural Information Processing Systems}.

\bibitem[Tsur et~al.(2025{\natexlab{a}})Tsur, Long, Verdun, Hsu, Chen, Permuter, Vithana, and Calmon]{ot_watermark}
Tsur, D., Long, C.~X., Verdun, C.~M., Hsu, H., Chen, C.-F., Permuter, H., Vithana, S., and Calmon, F.~P.
\newblock Heavywater and simplexwater: Watermarking low-entropy text distributions.
\newblock \emph{arXiv preprint arXiv:2506.06409}, 2025{\natexlab{a}}.

\bibitem[Tsur et~al.(2025{\natexlab{b}})Tsur, Long, Verdun, Hsu, Permuter, and Calmon]{optimized_couplings_wm}
Tsur, D., Long, C.~X., Verdun, C.~M., Hsu, H., Permuter, H., and Calmon, F.~P.
\newblock Optimized couplings for watermarking large language models.
\newblock \emph{arXiv preprint arXiv:2505.08878}, 2025{\natexlab{b}}.

\bibitem[Wouters(2023)]{optimizing_wm}
Wouters, B.
\newblock Optimizing watermarks for large language models.
\newblock \emph{arXiv preprint arXiv:2312.17295}, 2023.

\bibitem[Wu et~al.(2023)Wu, Hu, Zhang, and Huang]{dipmark}
Wu, Y., Hu, Z., Zhang, H., and Huang, H.
\newblock Dipmark: A stealthy, efficient and resilient watermark for large language models.
\newblock 2023.

\bibitem[Wu et~al.(2025)Wu, Cui, Chen, and Huang]{wu2025analyzing}
Wu, Y., Cui, X., Chen, R., and Huang, H.
\newblock Analyzing and evaluating unbiased language model watermark.
\newblock \emph{arXiv preprint arXiv:2509.24048}, 2025.

\bibitem[Zhou et~al.(2024)Zhou, Zhao, Xu, and Ren]{bileve}
Zhou, T., Zhao, X., Xu, X., and Ren, S.
\newblock Bileve: Securing text provenance in large language models against spoofing with bi-level signature.
\newblock \emph{Advances in Neural Information Processing Systems}, 37:\penalty0 56054--56075, 2024.

\end{thebibliography}
